\documentclass[12pt, draftclsnofoot, onecolumn]{IEEEtran}

\usepackage{cite}

\usepackage{graphicx} 
\usepackage{subcaption}

\usepackage[cmex10]{amsmath}
\usepackage{amssymb}
\usepackage{color}
\usepackage{amsthm}

\newtheorem{theorem}{Theorem}
\newtheorem{lemma}{Lemma}
\newtheorem{corollary}{Corollary}
\newtheorem{proposition}{Proposition}

\newtheorem{remark}{Remark}

\DeclareMathOperator*{\argmin}{arg\,min}

\begin{document}

\title{Max-Min Fair Transmit Precoding for Multi-group Multicasting in Massive MIMO}

\author{\small{Meysam Sadeghi, Student Member, IEEE\thanks{M. Sadeghi (meysam@mymail.sutd.edu.sg) and C. Yuen (yuenchau@sutd.edu.sg) are with Singapore University of Technology and Design, Singapore.
		E. Bj\"{o}rnson (emil.bjornson@liu.se) and E. G. Larsson (erik.g.larsson@liu.se) are with Department of Electrical Engineering (ISY), Link\"{o}ping University, Link\"{o}ping, Sweden.
		T. L. Marzetta (tom.marzetta@nokia-bell-labs.com) is with Bell Laboratories, Alcatel-Lucent, 600 Moutain Avenue, Murray Hill, NJ, 07974, USA. 
		This work was supported in part by the Swedish Research Council (VR), the Swedish Foundation for Strategic Research (SSF), and ELLIIT.
		\newline \indent Parts of this paper has been submitted to IEEE GLOBECOM 2017 \cite{GC2017}. In \cite{GC2017}, we briefly present ZF-undp and ZF-mudp, refer to Fig. \ref{figint}. In this manuscript, we present a detailed description of all six scenarios of Fig. \ref{figint}. },
	Emil Bj\"{o}rnson, Senior Member, IEEE,
	Erik~G.~Larsson, Fellow, IEEE,
	Chau~Yuen, Senior Member, IEEE,
	and~Thomas~L.~Marzetta, Fellow, IEEE}
	\vspace{-0.8cm}}

\maketitle

\begin{abstract}
This paper considers the downlink precoding for physical layer multicasting in massive multiple-input-multiple-output (MIMO) systems. We study the max-min fairness (MMF) problem, where channel state information (CSI) at the transmitter is used to design precoding vectors that maximize the minimum spectral efficiency (SE) of the system, given fixed power budgets for uplink training and downlink transmission. Our system model accounts for channel estimation, pilot contamination, arbitrary path-losses, and multi-group multicasting. We consider six scenarios with different transmission technologies (unicast and multicast), different pilot assignment strategies (dedicated or shared pilot assignments), and different precoding schemes (maximum ratio transmission and zero forcing), and derive achievable spectral efficiencies for all possible combinations. Then we solve the MMF problem for each of these scenarios and for any given pilot length we find the SE maximizing uplink pilot and downlink data transmission policies, all in closed-forms. We use these results to draw a general guideline for massive MIMO multicasting design, where for a given number of base station antennas, number of users, and coherence interval length, we determine the multicasting scheme that shall be used.
\end{abstract}


\section{Introduction}
The advent of smartphones and tablets, data hungry applications, and the ever growing amount of digital content have increased the mobile data traffic unprecedentedly \cite{Cisco2016}. It is anticipated that the mobile data traffic will grow at a compound annual growth rate of $53 \%$ from $2015$ to $2020$ and reach $30.6$ exabytes per month \cite{Cisco2016}. A considerable portion of this traffic belongs to the contents that are of interest for groups of users in the network, for example, live broadcast of sporting events, mobile TV, and regular system updates  \cite{lecompte2012evolved,DVB-GFaria,DVB-Elhajjar}. Although these types of traffic can be delivered by unicast\footnote{We present a formal definition of unicast and multicast transmissions in Section II. B.} transmission, theoretically it is more efficient to employ multicast transmission\footnote{For the sake of brevity, in this paper we refer to physical layer multicasting as multicasting.} \cite{sidiropoulos2006transmit} and therefore it has been considered in different releases of the 3rd Generation Partnership Project \cite{lecompte2012evolved}.

Multicasting can be performed in two different ways, either with the blind isotropic transmission as in digital video broadcasting \cite{DVB-GFaria,DVB-Elhajjar} or by downlink precoding based on channel state information (CSI) \cite{sidiropoulos2006transmit,karipidis2008quality}. As detailed in \cite{karipidis2008quality}, the latter approach is more desirable for wireless systems. In this paper by multicasting we refer to this second approach, where the multi-antenna transmitter employs its CSI to perform precoding such that a desired metric of interest is optimized \cite{sidiropoulos2006transmit,karipidis2008quality}. A seminal study of multicasting is presented in \cite{sidiropoulos2006transmit}, where the precoder design for the so-called max-min fairness (MMF) and quality of service (QoS) problems is investigated. Considering a single-group single-cell system, it is shown that both MMF and QoS problems are NP-hard and a suboptimal solution is presented. This work is then extended to a multi-group single-cell scenario and it is shown that there exists a duality between the MMF and QoS problems \cite{karipidis2008quality}. The MMF problem is then revisited under per-antenna power constraint for multi-group single-cell systems in \cite{christopoulos2014weighted}. Also, the coordinated multicasting transmission for a single-group multi-cell scenario is investigated in \cite{xiang2013coordinated}. Note that \cite{sidiropoulos2006transmit,karipidis2008quality,christopoulos2014weighted,xiang2013coordinated} assume perfect CSI is available at the base station (BS) and also at the user terminals (UTs).

The aforementioned works (among many others) are based on the semidefinite relaxation (SDR) technique and suffer from high computational complexity. Considering a multicasting system with an $N$-antenna BS and $G$ different multicasting groups, the complexity of SDR based techniques is of $\mathcal{O}(G^{3.5}N^{6.5})$ \cite{karipidis2008quality}. This high complexity makes the SDR based multicasting algorithms impractical for large dimensional systems, e.g. massive MIMO systems where they deploy hundreds of antennas \cite{marzetta2010noncooperative}.

Due to significant performance of massive MIMO in terms of energy and spectral efficiency \cite{hoydis2013massive,ngo2013energy,LSAPowerNorm}, it is a promising candidate for the fifth generation of cellular networks \cite{andrews2014will,boccardi2014five}. Therefore recent works on multicasting have tried to address the high computational complexity of massive MIMO multicasting \cite{tran2014conic,christopoulos2015multicast,MeysamMultiComplexity}. Particularly, \cite{tran2014conic} presents a successive convex approximation technique for single-group single-cell multicasting of large-scale antenna arrays which reduces the computational complexity to $\mathcal{O}(N^{3.5})$. The system set-up of \cite{tran2014conic} is extended to a  multi-group single-cell multicasting in \cite{christopoulos2015multicast}. Therein a feasible point pursuit based algorithm with a complexity of $\mathcal{O}((GN)^{3.5})$ is presented. However, the complexity is still high for large-scale antenna systems with hundreds of antennas. Recently a low-complexity algorithm, $\mathcal{O}(N)$ for single-group and $\mathcal{O}(GN^{2})$ for multi-group multicasting, for massive MIMO system is presented in \cite{MeysamMultiComplexity}. This algorithm not only reduces the complexity but also significantly outperforms the SDR based methods.

The common denominator of the aforementioned algorithms is the perfect CSI assumption, both at the BS and at the UTs. However, in practice the CSI is not available neither at the BS nor at the UTs, and should be obtained. This introduces new challenges to the multicasting problem, which is already NP-hard. To address the CSI acquisition problem, two approaches have been presented in the literature. The first approach leverages the asymptotic orthogonality of the channels in massive MIMO, which simplifies the precoding design \cite{Zhengzheng2014,zhou2015joint,sadeghi2015multi}. The main problem with the asymptotic approach is that a very large number of antennas, e.g.,  $N>4000$, is required to get close to the asymptotic performance, while the performance is poor for realistic antenna numbers \cite{zhou2015joint,sadeghi2015multi}.

The second approach relies on employing predefined multicasting precoders \cite{YangMulticat}. More precisely, considering a single-cell multi-group multicasting system, \cite{YangMulticat} presents a maximum ratio transmission (MRT) based multicast precoder with a novel pilot allocation strategy. Contrary to the common approach where a dedicated pilot is used per UT, it uses a shared pilot for all the UTs within each multicasting group, hereafter called co-pilot assignment. They show numerically that MRT multicasting with co-pilot assignment substantially outperforms the MRT unicasting with dedicated pilot assignment in terms of minimum spectral efficiency (SE). 

The improvement in the SE of multicast transmission, shown by \cite{YangMulticat}, has motivated the application of co-pilot assignment in the subsequent works \cite{Zhengzheng2014,zhou2015joint,sadeghi2015multi}. But as this improved SE is observed by numerical comparison of MRT multicast transmission with co-pilot assignment, and MRT unicast transmission with dedicated pilot assignment, a series of questions remain to be answered:
\begin{itemize}
	\item Does the same observation hold for zero forcing (ZF)?
	\item When is beneficial to employ co-pilot assignment instead of dedicated pilot assignment? 
	\item Given a set of system parameters, which precoder and pilot assignment shall be used?  
\end{itemize}
To answer these questions, we study six different possible scenarios as shown in Fig. \ref{figint}. The first layer  of Fig. \ref{figint} considers the two possible transmission technology, unicast (un) and multicast (mu). The second layer considers the employed pilot assignment strategy\footnote{Note that for unicast we just consider dedicated pilot assignment as the co-pilot assignment results in very weak performance due to high inter-group interference and extreme pilot contamination.}, i.e., dedicated pilot (dp) or co-pilot (cp). The third layer determines the precoding scheme, which is either MRT or ZF. Then the six considered scenarios are: MRT-undp, ZF-undp, MRT-mudp, ZF-mudp, MRT-mucp, and ZF-mucp.\footnote{As an example note that MRT-mucp means MRT multicasting with co-pilot assignment.} 

\begin{figure}[]
	\centering
	\includegraphics[width=1\columnwidth, trim={0.25cm 2.9cm 0.2cm 2.8cm},clip]{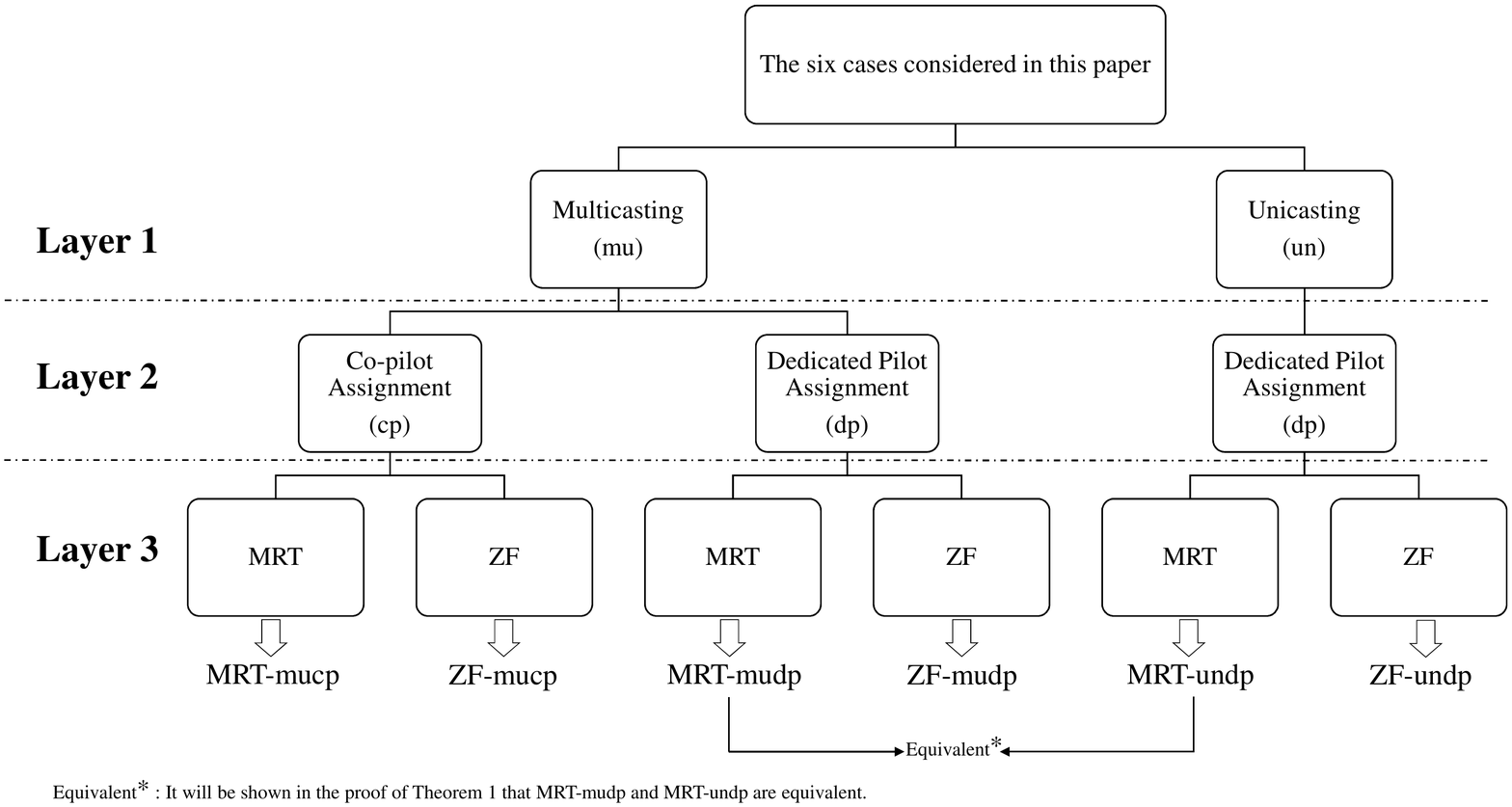}
	\caption{The six considered scenarios in this paper.}
	\label{figint}
\end{figure}

In this paper, we answer the aforementioned questions while considering a multi-group massive MIMO multicasting system with realistic CSI acquisition. Our main contributions are as follows:
\begin{itemize}
	\item We derive achievable SEs for each UT in the system considering the set-ups depicted in Fig. \ref{figint}.
	
	\item We formulate the MMF problem for each of the six scenarios in Fig. \ref{figint}. For an arbitrary pilot length, we find 1) the optimal uplink pilots powers; 2) the optimal downlink data transmission powers; and 3) the optimal SE for each UT in the system, all in closed-forms.
	
	\item Based on our analytical and numerical results, we draw a guideline for massive MIMO multicasting design. More precisely, given the number of BS antennas, the number of UTs, and the length of coherence interval, we determine the multicasting scheme that shall be used.
\end{itemize}

The remainder of this paper is organized as follows. Section II introduces the system model, the channel estimation, and elaborates the unicast and multicast transmissions. Section III presents the precoding schemes and their associated achievable SEs. Section IV studies the MMF problem for all set-ups of Fig. \ref{figint}. Section V presents the numerical analysis and further detailed discussions. Section VI summarizes the paper and presents the main conclusions.

\textit{Notations:} The following notation is used throughout the paper. Scalars are denoted by lower
case letters whereas boldface lower (upper) case letters are used for vectors (matrices). We
denote by $\mathbf{I}_{G}$ the identity matrix of size $G$ and represent the $j$ column of $\mathbf{I}_{G}$ as $\mathbf{e}_{j,G}$. The symbol $\mathcal{CN} (.,.)$ denotes the circularly symmetric complex Gaussian distribution. The trace, transpose, conjugate transpose, and expectation operators are denoted by $\mathrm{tr}(.)$, $(.)^{T}$ , $(.)^{H}$, and $\mathbb{E}[.]$, respectively. We denote the cardinality of a set $\mathcal{G}$ by $\vert \mathcal{G} \vert$. 

\section{System and Signal Model}
We consider multi-group multicasting in a single-cell massive MIMO system. We assume the system has one BS with $N$ antennas and it transmits $G$ data streams toward $G$ multicasting groups. We denote the set of indices of these $G$ multicasting groups as $\mathcal{G}$, i.e., $\mathcal{G} = \{ 1, \ldots , G \}$. We assume the $j$th data stream, $j \in \{1,\ldots,G\}$, is of interest for $K_{j}$ single antenna UTs, and we say these $K_{j}$ UTs belong to the $j$th multicasting group. We denote the set of indices of all the UTs in $j$th multicasting group as $\mathcal{K}_{j}$, i.e. $\mathcal{K}_{j} = \{ 1, \ldots , K_{j} \}$. Therefore $\vert \mathcal{G} \vert = G$ and $\vert \mathcal{K}_{j} \vert = K_{j}$. We assume each UT is assigned to just one multicasting group, i.e. $\mathcal{K}_{i} \cap \mathcal{K}_{j} =\emptyset  \; \forall i,j \in \mathcal{G}, i \neq j$. We denote the total number of UTs in the system as $K_{tot} = \sum_{j=1}^{G} K_{j}$.

We consider a block flat-fading channel model where $C_{B}$ (in Hz) is the coherence bandwidth and $C_{T}$ (in seconds) is the coherence time. Hence the channels are static within a coherence interval of $T=C_{B} C_{T}$ symbols. We assume the BS does not have a priori CSI but estimates the channels by uplink pilots transmission using a TDD protocol, exploiting channel-reciprocity. The procedure is detailed next. Under these assumptions, we represent the channel between the BS and UT $k$ in multicasting group $j$ as $\mathbf{g}_{jk}$. We assume all the UTs have independent Rayleigh fading channels, as it well-matches non-line-of-sight measurements \cite{XGaoMeasurment}. This implies that $\mathbf{g}_{jk} \sim \mathcal{CN}(\mathbf{0}, \beta_{jk} \mathbf{I}_{N}) \forall k, j$, where $\beta_{jk}$ represents the large-scale fading. 
\subsection{Channel Estimation}
The BS uses uplink pilot transmission to estimate the channels to the UTs in the system. As detailed in Section I, it can be performed either by dedicated pilot assignment \cite{ngo2013energy,hoydis2013massive}, or by co-pilot assignment \cite{YangMulticat}. The dedicated pilot approach sacrifices more resources, e.g., time-frequency slots in each coherence interval, to achieve a better estimation of the channel of each UT in the system. On the other hand, the co-pilot approach enforces deliberate pilot contamination among UTs of each multicasting group in order to reduce the consumed time-frequency resources. In the sequel we elaborate the channel estimation under each of these scenarios.

\subsubsection{Channel Estimation with Dedicated Pilot Assignment}
The dedicated pilot assignment uses one pilot per UT, so it requires $K_{tot}$ pilots per coherence interval. Denoting the pilot length as $\tau_{p}^{dp}$, to have orthogonal pilots we have $\tau_{p}^{dp} \geq K_{tot}$. Under dedicated pilot assignment, the minimum mean-square error (MMSE) estimate of the channel of UT $k$ in group $j$ is
\begin{align}
\label{est_gjk_dp}
\hat{\mathbf{g}}_{jk}^{dp} = \dfrac{\sqrt{\tau_{p}^{dp} p_{jk}^{u}} \beta_{jk} }{1 + \tau_{p} p_{jk}^{u} \beta_{jk}} \left(  \sqrt{\tau_{p}^{dp} p_{jk}^{u}}  \mathbf{g}_{jk} + \mathbf{n} \right) 
\end{align}
where $ \mathbf{n} \sim \mathcal{CN}(\mathbf{0}, \mathbf{I}_{N})$ is the normalized additive noise and $p_{jk}^{u}$ is the uplink pilot power of UT $k$ in group $j$. Therefore we have $\hat{\mathbf{g}}_{jk}^{dp} \sim \mathcal{CN} ( \mathbf{0}, \gamma_{jk}^{dp} \mathbf{I}_{N})$ with $\gamma_{jk}^{dp} =  \dfrac{\tau_{p}^{dp} p_{jk}^{u} \beta_{jk}^{2} }{1 + \tau_{p}^{dp} p_{jk}^{u} \beta_{jk}} $.  Also the estimation error is $\tilde{\mathbf{g}}_{jk}^{dp} = \hat{\mathbf{g}}_{jk}^{dp} - \mathbf{g}_{jk}  \sim \mathcal{CN}(\mathbf{0}, (\beta_{jk} - \gamma_{jk}^{dp}) \mathbf{I}_{N}) $. Moreover, we denote the $N \times K_{tot}$ matrix obtained by stacking the estimated channel of all the UTs in the system as $\hat{\mathbf{G}}_{dp} = [\hat{\mathbf{G}}_{1},\ldots,\hat{\mathbf{G}}_{G}]$, where $ \hat{\mathbf{G}}_{j} = [\hat{\mathbf{g}}_{j1}^{dp},\ldots,\hat{\mathbf{g}}_{jK_{j}}^{dp}] \; \forall j \in \mathcal{G}$.

\subsubsection{Channel Estimation with Co-pilot Assignment}
The co-pilot assignment uses one pilot per multicast group, so it requires $G$ pilots per coherence interval. Denoting the pilot length as $\tau_{p}^{cp}$, to have orthogonal pilots we need $\tau_{p}^{cp} \geq G$. Under co-pilot assignment the MMSE estimate of the channel of UT $k$ in multicasting group $j$ is
\begin{align}
\label{est_gjk_cp}
\hat{\mathbf{g}}_{jk}^{cp} = \dfrac{\sqrt{\tau_{p}^{cp} p_{jk}^{u}} \beta_{jk}}{ 1 + \tau_{p}^{cp} \sum_{m=1}^{K_{j}} p_{jm}^{u} \beta_{jm} } \left(  \sum_{m=1}^{K_j} \sqrt{\tau_{p}^{cp} p_{jm}^{u}} \mathbf{g}_{jm}  + \mathbf{n} \right) 
\end{align} 
where $\hat{\mathbf{g}}_{jk}^{cp} \sim \mathcal{CN}(\mathbf{0}, \gamma_{jk}^{cp}  \mathbf{I}_{N})$ with $\gamma_{jk}^{cp} = \dfrac{\tau_{p}^{cp} p_{jk}^{u} \beta_{jk}^{2}}{ 1 + \tau_{p}^{cp} \sum_{m=1}^{K_{j}} p_{jm}^{u} \beta_{jm}}$. From \eqref{est_gjk_cp} it is easy to observe that the channel estimate of each UT is contaminated by other UTs in its multicasting group. The estimation error of $\mathbf{g}_{jk}$ is $\tilde{\mathbf{g}}_{jk}^{cp} = \hat{\mathbf{g}}_{jk}^{cp} - \mathbf{g}_{jk} \sim \mathcal{CN}(\mathbf{0}, (\beta_{jk} - \gamma_{jk}^{cp}) \mathbf{I}_{N}) $. Moreover, we need the estimation of a linear combination of the channels of all the UTs within this multicasting group, which we denote as $\mathbf{g}_{j} = \sum_{k=1}^{K_j} \sqrt{\tau_{p}^{cp} p_{jk}^{u}} \mathbf{g}_{jk}$.\footnote{Note that for $K_{j}=1$, $\mathbf{g}_{j} = \sqrt{\tau_{p}^{cp} p_{jk}^{u}} \mathbf{g}_{jk}$.} Its MMSE estimate is 
\begin{align}
\label{estgj}
\hat{\mathbf{g}}_{j} = \dfrac{ \tau_{p}^{cp} \sum_{k=1}^{K_{j}} p_{jk}^{u} \beta_{jk}}{1 + \tau_{p}^{cp} \sum_{k=1}^{K_{j}} p_{jk}^{u} \beta_{jk}}  \left( \sum_{k=1}^{K_j} \sqrt{\tau_{p}^{cp} p_{jk}^{u}} \mathbf{g}_{jk}  + \mathbf{n} \right)                                                                                                       
\end{align}
and we have $\hat{\mathbf{g}}_{j} \sim  \mathcal{CN}(\mathbf{0},  \gamma_{j} \mathbf{I}_{N})$ with $\gamma_{j} = \dfrac{ (\tau_{p}^{cp} \sum_{k=1}^{K_{j}} p_{jk}^{u} \beta_{jk})^2}{1 + \tau_{p}^{cp} \sum_{k=1}^{K_{j}} p_{jk}^{u} \beta_{jk}}$. Also we denote the $N \times G$ matrix obtained by stacking the vectors $\hat{\mathbf{g}}_{j}$ $\forall j \in \mathcal{G}$ as $\hat{\mathbf{G}}_{cp} = [\hat{\mathbf{g}}_{1},\ldots,\hat{\mathbf{g}}_{{G}}]$.

\subsection{Transmission Mode: Unicast versus Multicast}
As motivated in Section I, we want to understand when it is beneficial to employ multicast transmission instead of unicast transmission. Therefore we consider both unicast and multicast transmissions in the sequel. Let us denote by $s_{i} \sim \mathcal{CN}(0,1) \; \forall i \in \mathcal{G}$ the signal requested by the UTs in the $i$th multicasting group, i.e., $\mathcal{K}_{i}$. We assume $s_{i}$ is independent across $i$. We stack them in a vector $\mathbf{s} = [s_{1}, \ldots, s_{G}]^{T}$.

\subsubsection{Unicast Transmission}
In unicast transmission we consider a $K_{tot} \times 1$ data vector $\mathbf{x}$ where
\begin{align}
\mathbf{x} = [\underbrace{s_{1},\ldots,s_{1}}_{K_{1}}, \underbrace{s_{2},\ldots,s_{2}}_{K_{2}}, \ldots, \underbrace{s_{G},\ldots,s_{G}}_{K_{G}}]^{T}.
\end{align}
Also the precoding matrix is an $N \times K_{tot}$ matrix $\mathbf{W}_{un} = [\mathbf{w}_{11},\ldots,\mathbf{w}_{jk}, \ldots, \mathbf{w}_{GK_{G}}]$, where $\mathbf{w}_{jk}$ is the precoding vector of UT $k$ in multicasting group $j$. We will provide more details on the exact structure of the precoding vectors in Section II. The received signal of UT $k$ in multicasting group $j$ during downlink transmission is 
\begin{align}
\label{unicasttransmission}
y_{jk} = \mathbf{g}_{jk}^{H} \mathbf{W}_{un} \mathbf{x} + n =  \mathbf{g}_{jk}^{H} \sum_{i=1}^{G} \sum_{t=1}^{K_{i}} \mathbf{w}_{it} s_{i} + n
\end{align}
where $n \sim \mathcal{CN}(0, 1)$ is the normalized noise.

\subsubsection{Multicast Transmission}
In the multicast case we use $\mathbf{s}$ as the data vector. Also the precoding matrix becomes an $N \times G$ matrix $\mathbf{W}_{mu} = [\mathbf{w}_{1},\ldots, \mathbf{w}_{G}]$ where $\mathbf{w}_{j}$ is the joint precoding vector of all the UTs in $j$th multicasting group. In this case, the received signal of UT $k$ in multicasting group $j$ is
\begin{align}
\label{multicasttransmission}
y_{jk} = \mathbf{g}_{jk}^{H} \mathbf{W}_{mu} \mathbf{s} + n =  \mathbf{g}_{jk}^{H} \sum_{i=1}^{G} \mathbf{w}_{i} s_{i} + n.
\end{align}

\section{Precoder Structures and Achievable SEs}
It is well known that in massive MIMO systems linear precoding schemes provide close-to-optimal performance \cite{Emil10Myth}. Also it has been shown that the asymptotically optimal precoders in massive MIMO multicasting are linear combinations of the channels \cite{sadeghi2015multi,Zhengzheng2014}. Therefore, in the sequel we consider two common linear precoding schemes in the context of massive MIMO systems, namely MRT and ZF \cite{yang2013performance}, and derive the achievable SE for them.


\subsection{Precoder Structure and Achievable SE for Unicast Transmission} 
Consider the $N \times K_{tot}$ precoding matrix $\mathbf{W}_{un} = [\mathbf{w}_{11},\ldots,\mathbf{w}_{jk}, \ldots, \mathbf{w}_{GK_{G}}]$ for unicast transmission with dedicated pilots. Then the MRT and ZF precoding vectors of UT $k$ in group $j$ are
\begin{align}
\label{MRTUNDP}
& \mathbf{w}_{jk}^{\rm{MRT-undp}} = \sqrt{\dfrac{p_{jk}^{dl}}{N \gamma_{jk}^{dp}}} \; \hat{\mathbf{g}}_{jk}^{dp}
\\
\label{ZFUNDP}
& \mathbf{w}_{jk}^{\rm{ZF-undp}} = \sqrt{p_{jk}^{dl} \gamma_{jk}^{dp} ( N-K_{tot})} \;\; \hat{\mathbf{G}}_{dp} (\hat{\mathbf{G}}^{H}_{dp} \hat{\mathbf{G}}_{dp})^{-1} \mathbf{e}_{\nu_{jk},K_{tot}}
\end{align}
where  $\nu_{jk}=\sum_{t=1}^{j-1} K_{t} + k$, $\mathbf{e}_{\nu_{jk},K_{tot}}$ is the $\nu_{jk}$th column of $\mathbf{I}_{K_{tot}}$, and $p_{jk}^{dl}$ is the downlink power allocated to this user. Note that for $\mathbf{w}_{jk}^{\rm{MRT}}$ and $\mathbf{w}_{jk}^{\rm{ZF}} $, we have $\mathbb{E}[\Vert \mathbf{w}_{jk}^{\rm{MRT}} \Vert^{2}] = p_{jk}^{dl}$ and $\mathbb{E}[\Vert \mathbf{w}_{jk}^{\rm{ZF}} \Vert^{2}] = p_{jk}^{dl}$. We denote the total utilized downlink power as $P_{dp}=\sum_{j=1}^{G} \sum_{k=1}^{K_{j}} p_{jk}^{dl}$. Given \eqref{MRTUNDP} and \eqref{ZFUNDP}, we can achieve the following SEs for the UTs in the system. 

\begin{proposition} \label{prop1}
	With MRT unicast transmission and dedicated pilot assignment, an achievable SE for user $k$ of group $i$ is
	\begin{align}
	\label{se_mrt_undp}
	\mathrm{SE}_{ik}^{\mathrm{MRT-undp}} = \left( 1 - \dfrac{\tau_{p}^{dp}}{T}\right)  \log_{2} (1+\mathrm{SINR}_{ik}^{\mathrm{MRT-undp}})
	\end{align}
	where $\mathrm{SINR}_{ik}^{\mathrm{MRT-undp}} = \dfrac{N  \gamma_{ik}^{dp} p_{ik}^{dl} }{ 1 +  \beta_{ik} P_{dp}}$ is the effective SINR of this user.
\end{proposition}
\begin{proof}
	The proof follows the conventional bounding\footnote{In Propositions 1 and 2, the achievable SE is obtained by employing the use and then forget (UatF) bounding technique\cite{marzetta2016fundamentals,jose2011pilot}. Compared to the classic application of UatF in massive MIMO, here we have a subtle technicality as follows. The interference caused by the transmission to the other UTs in group $i$ is uncorrelated with the effective transmission to user $k$ in group $i$, however the message is the same. Therefore the transmission to the other UTs within a multicast group does not contribute to the desired signal power and act as interference.} technique in \cite{marzetta2016fundamentals} and is omitted for brevity.
\end{proof}


\begin{proposition}
	With ZF unicast transmission and dedicated pilot assignment, an achievable SE for user $k$ of group $i$ is
	\begin{align}
	\label{se_zf_undp}
	\mathrm{SE}_{ik}^{\mathrm{ZF-undp}} = \left( 1 - \dfrac{\tau_{p}^{dp}}{T} \right)  \log_{2} (1+\mathrm{SINR}_{ik}^{\mathrm{ZF-undp}})
	\end{align}
	where $\mathrm{SINR}_{ik}^{\mathrm{ZF-undp}} = \dfrac{(N-K_{tot}) \gamma_{ik}^{dp} p_{ik}^{dl} }{1 + (\beta_{ik} - \gamma_{ik}^{dp}) P_{dp}}$ is the effective SINR of this user.
\end{proposition}
\begin{proof}
	The proof follows the conventional bounding technique in \cite{marzetta2016fundamentals} and is omitted for brevity.
\end{proof}


\subsection{Precoder Structure and Achievable SEs for Multicast Transmission}
As detailed in Section II.A, the required CSI for multicast transmission can be achieved either by dedicated pilot assignment or by co-pilot assignment. In the sequel we present the precoder structure and achievable SEs for both cases.


\subsubsection{Precoder Structure and Achievable SE for Multicast Transmission with Dedicated Pilot Assignment}
If dedicated pilot assignment is employed then the MRT and ZF precoding vectors of $j$th multicast group are
\begin{align}
\label{MRTMUDP}
& \mathbf{w}_{j}^{\rm{MRT-mudp}} = \sum_{k=1}^{K_{j}} \sqrt{\dfrac{p_{jk}^{dl}}{N \gamma_{jk}^{dp}}} \; \hat{\mathbf{g}}_{jk}^{dp}
\\
\label{ZFMUDP}
& \mathbf{w}_{j}^{\rm{ZF-mudp}} = (\mathbf{I}_{N} - \hat{\mathbf{G}}_{-j} (\hat{\mathbf{G}}_{-j}^{H} \hat{\mathbf{G}}_{-j})^{-1} \hat{\mathbf{G}}_{-j}^{H} ) \sum_{k=1}^{K_{j}} \sqrt{\mu_{jk}} \hat{\mathbf{g}}_{jk}^{dp} 
\end{align} 
where $p_{jk}^{dl}$ is the downlink power of UT $k$ in group $j$, $\hat{\mathbf{G}}_{-j} = [\hat{\mathbf{G}}_{1}, \ldots,\hat{\mathbf{G}}_{j-1},\hat{\mathbf{G}}_{j+1},\ldots,\hat{\mathbf{G}}_{G}]$, and $\mu_{jk} = \sqrt{\dfrac{p_{jk}^{dl}}{(N - \nu_{j}) \gamma_{jk}^{dp}}}$ with $\nu_{j} = K_{tot} - K_{j}$. For $\mathbf{w}_{j}^{\rm{MRT-mudp}}$ and $\mathbf{w}_{j}^{\rm{ZF-mudp}}$ we have $\mathbb{E}[\Vert \mathbf{w}_{j}^{\rm{MRT-mudp}} \Vert^{2}] = \sum_{k=1}^{K_{j}} p_{jk}^{dl}$ and $\mathbb{E}[\Vert \mathbf{w}_{j}^{\rm{ZF-mudp}} \Vert^{2}] = \sum_{k=1}^{K_{j}} p_{jk}^{dl}$.

Note that there is a subtle difference between ZF-undp and ZF-mudp. The ZF-undp scheme ensures that (within the limitations of channel estimation errors) any UT is immune to the transmissions intended for all other UTs, in its own multicasting group and also in other multicasting groups. Therefore it requires $N \geq K_{tot}$. However, ZF-mudp just ensures that the UTs within each multicasting group are rendered immune (within the limitations of channel estimation errors) to the transmissions to the rest of UTs in other multicasting groups and every UT experiences intra-group interference from the transmissions intended for the other UTs in its own group. Hence it requires $N \geq (K_{tot} - \max_{j \in \mathcal{G}} K_{j})$.

\begin{remark}
	\label{remZFmudp}
	Notice that \eqref{ZFMUDP} is a generalized version of the precoder proposed in \cite{MeysamMultiComplexity}, since that it accounts for imperfect CSI. As the precoder presented in \cite{MeysamMultiComplexity} outperforms the SDR based multicasting schemes, this generalization works as a benchmark and enable us to indirectly compare our proposed methods with the SDR based algorithms. This is of particular interest, as the SDR-based algorithms, which are assuming perfect CSI is available at both BS and UTs, are the baseline schemes used in the literature \cite{sidiropoulos2006transmit,karipidis2008quality,christopoulos2014weighted,xiang2013coordinated}.
\end{remark}

Given \eqref{MRTMUDP} and \eqref{ZFMUDP}, we can achieve the following SEs.

\begin{theorem}
	\label{T-MRT-mudp}
	With MRT multicast transmission and dedicated pilot assignment, an achievable SE for user $k$ of group $i$ is
	\begin{align}
	\label{se_mrt_mudp}
	\mathrm{SE}_{ik}^{\mathrm{MRT-mudp}} = \left( 1 - \dfrac{\tau_{p}^{dp}}{T}\right)  \log_{2} (1+\mathrm{SINR}_{ik}^{\mathrm{MRT-mudp}}).
	\end{align}
	where $\mathrm{SINR}_{ik}^{\mathrm{MRT-mudp}} = \dfrac{N \gamma_{ik}^{dp} p_{ik}^{dl}}{1 +  \beta_{ik} P_{dp}}$ is the effective SINR of this user.
\end{theorem}
\begin{proof}
	The proof follows by showing that when we have a common message for all the users in each multicasting group, the MRT-mudp is equivalent with MRT-undp: 
	\begin{align*}
	\mathbf{W}_{un} \mathbf{x} = \sum_{j=1}^{G} \sum_{k=1}^{K_{j}} \mathbf{w}_{jk}^{\mathrm{MRT-undp}} s_{j} =  \sum_{j=1}^{G} \mathbf{w}_{j}^{\mathrm{MRT-mudp}} s_{j} = \mathbf{W}_{mu} \mathbf{s}.
	\end{align*}
	Hence the SINR and SE are the same as Proposition \ref{prop1}.
\end{proof}

\begin{theorem}
	\label{TZFmudp}
	With ZF multicast transmission and dedicated pilot assignment, an achievable SE for user $k$ of group $i$ is
	\begin{align}
	\label{se_zf_mudp}
	\mathrm{SE}_{ik}^{\mathrm{ZF-mudp}} = \left( 1 - \dfrac{\tau_{p}^{dp}}{T}\right)  \log_{2} (1+\mathrm{SINR}_{ik}^{\mathrm{ZF-mudp}}).
	\end{align}
	where $\mathrm{SINR}_{ik}^{\mathrm{ZF-mudp}} = \dfrac{(N - \nu_{i}) \gamma_{ik}^{dp} p_{ik}^{dl}}{1+ \gamma_{ik}^{dp} \sum_{m=1}^{K_{i}} p_{im}^{dl}   +  (\beta_{ik} - \gamma_{ik}^{dp}) P_{dp}}$ is the effective SINR of this user.
\end{theorem}
\begin{proof}
	The proof is given in Appendix A.
\end{proof}


\subsubsection{Precoder Structure for Multicast Transmission with Co-pilot Assignment}
If co-pilot assignment is utilized then the MRT and ZF precoding vectors of $j$th multicast group are
\begin{align}
\label{MRTMUCP}
& \mathbf{w}_{j}^{\rm{MRT-mucp}} = \sqrt{\dfrac{p_{j}^{dl}}{N \gamma_{j}}} \; \hat{\mathbf{g}}_{j}
\\
\label{ZFMUCP}
& \mathbf{w}_{j}^{\rm{ZF-mucp}} = \sqrt{p_{j}^{dl} \gamma_{j} ( N-G)} \;\; \hat{\mathbf{G}}_{cp} (\hat{\mathbf{G}}_{cp}^{H} \hat{\mathbf{G}}_{cp})^{-1} \mathbf{e}_{j,G} 
\end{align} 
where $p_{j}^{dl}$ is the downlink power of the precoding vector of group $j$. Note that for $\mathbf{w}_{j}^{\rm{MRT-mucp}} $ and $\mathbf{w}_{j}^{\rm{ZF-mucp}}$ we have $\mathbb{E}[\Vert \mathbf{w}_{j}^{\rm{MRT-mucp}}  \Vert^{2}] =  p_{j}^{dl}$ and $\mathbb{E}[\Vert \mathbf{w}_{j}^{\rm{ZF-mucp}} \Vert^{2}] =  p_{j}^{dl}$. We denote the utilized downlink power as $P_{cp} = \sum_{j=1}^{G} p_{j}^{dl}$. By using MRT as in \eqref{MRTMUCP}, it has been shown that the following achievable SE for user $k$ of group $i$ can be obtained \cite{YangMulticat}
\begin{align}
\label{se_mrt_mucp}
\mathrm{SE}_{ik}^{\mathrm{MRT-mucp}} = \left( 1 - \dfrac{\tau_{p}^{cp}}{T} \right)  \log_{2} (1+\mathrm{SINR}_{ik}^{\mathrm{MRT-mucp}})
\end{align}
where $\label{sinr_mrt_mucp} \mathrm{SINR}_{ik}^{\mathrm{MRT-mucp}} = \dfrac{ N \gamma_{ik}^{cp} p_{i}^{dl}}{1+\beta_{ik}P_{cp}}$ is the effective SINR of this user. By using ZF as in \eqref{ZFMUCP}, we can achieve the following SE.

\begin{theorem}
	\label{Theorem3}
	With ZF multicast transmission and co-pilot assignment, an achievable SE for user $k$ of group $i$ is
	\begin{align}
	\label{se_zf_mucp}
	\mathrm{SE}_{ik}^{\mathrm{ZF-mucp}} = \left( 1 - \dfrac{\tau_{p}^{cp}}{T}\right)  \log_{2} (1+\mathrm{SINR}_{ik}^{\mathrm{ZF-mucp}}).
	\end{align}
	where $	\label{sinr_zf_mucp} \mathrm{SINR}_{ik}^{\mathrm{ZF-mucp}} = \dfrac{(N-G) \gamma_{ik}^{cp} p_{i}^{dl} }{1+(\beta_{ik} - \gamma_{ik}^{cp})P_{cp}}$ is the effective SINR of this user.
\end{theorem}
\begin{proof}
	The proof is given in Appendix B.
\end{proof}

In Theorem \ref{Theorem3} we obtained a simple closed-form for the SINR of ZF-mucp, while the precoder is entirely based on the composite channels, e.g., $\hat{\mathbf{g}}_{j} \; \forall j \in \mathcal{G}$. This is because we took advantage of the fact that $\forall j \in \mathcal{G}, \forall k \in \mathcal{K}_{j}$, $ \hat{\mathbf{g}}_{jk}^{cp}$ and $ \hat{\mathbf{g}}_{j}$ are equal up to a scalar coefficient. Hence ZF-mucp can cancel the inter-group interference, within the limitation of the channel estimates, which leads to the obtained simple closed-form for the SINR of ZF-mucp. The proof details are given in Appendix B.

\begin{remark}
	\label{Rem-MRTtoZF}
	Note that when we switch from MRT to ZF in the above scenarios, e.g., from Proposition 1 to Proposition 2, the SINR terms always change in a particular way. The signal power in the numerator reduces by a factor of $\frac{N - \kappa}{N}$, where $\kappa$ depends on the considered scenario. Also the interference in the denominator reduces from $\beta_{ik} P_{dp}$ to $(\beta_{ik} - \gamma_{ik}^{dp}) P_{dp}$ or from $\beta_{ik} P_{cp}$ to $(\beta_{ik} - \gamma_{ik}^{cp}) P_{cp}$. This is due to the fact that ZF uses these $\kappa$ degrees of freedom to cancel the interference toward other UTs at the cost of reducing the received power of each UT. 
\end{remark}

\section{Max-Min Fairness Problem}
The MMF problem is the common problem of interest in multicasting systems, where we maximize the minimum of a metric of interest given some constraints on the resources. For the sake of simplicity, the existing works in the literature \cite{sidiropoulos2006transmit,karipidis2008quality,YangMulticat,MeysamMultiComplexity,Zhengzheng2014,christopoulos2015multicast,christopoulos2014weighted,xiang2013coordinated,zhou2015joint,sadeghi2015multi} consider the SINR as the metric of interest and the available power at the BS as the resource constraint, while ignoring CSI acquisition. Here we consider a more general problem formulation for MMF that accounts for the CSI acquisition. We choose the SE as our metric of interest and also we set our resource constraints as 1) the available power at the BS; 2) the uplink training power limit of the UTs; and 3) the length of the pilots. Therefore the MMF problem for dedicated pilot assignment is
\begin{align}
\label{MMF_dp_SE}
\mathcal{P}1: \max_{\tau_{p}^{dp}, \{p_{jk}^{dl}\}, \{p_{jk}^{u}\}} \min_{\forall j \in \mathcal{G}} & \min_{\forall k \in \mathcal{K}_{j}} \quad \; (1 - \frac{\tau_{p}^{dp}}{T}) \log_{2}(1 + \mathrm{SINR}_{jk}^{\mathrm{dp}})
\\
& s.t. \quad \quad p_{jk}^{u} \leq  p^{utot}_{jk} \quad \; \forall \; k \in \; \mathcal{K}_{j}, \forall \; j \in \; \mathcal{G} \tag{\ref{MMF_dp_SE}-C1}
\\
\label{poweryek1_SE}
&  \quad \quad  \quad \; P_{dp} = \sum_{j=1}^{G} \sum_{k=1}^{K_{j}} p_{jk}^{dl} \leq P	\tag{\ref{MMF_dp_SE}-C2}
\\
& \; \quad  \quad \quad \tau_{p}^{dp} \in \{K_{tot}, \ldots,T \} \tag{\ref{MMF_dp_SE}-C3}
\end{align}
where $p^{utot}_{jk}$ is the maximum pilot power of user $k$ in group $j$, and $P$ is the total available power at the BS. Similarly, the MMF problem for co-pilot assignment is

\begin{align}
\label{MMF_cp_SE}
\mathcal{P}2: \max_{\tau_{p}^{cp},\{p_{j}^{dl}\}, \{p_{jk}^{u}\}} \min_{\forall j \in \mathcal{G}} & \min_{\forall k \in \mathcal{K}_{j}} \quad (1 - \frac{\tau_{p}^{cp}}{T}) \log_{2}(1 + \mathrm{SINR}_{jk}^{\mathrm{cp}})
\\
& s.t. \quad \quad p_{jk}^{u} \leq  p^{utot}_{jk} \quad \; \forall \; k \in \; \mathcal{K}_{j}, \forall \; j \in \; \mathcal{G} \tag{\ref{MMF_cp_SE}-C1}
\\
\label{powerdo2_SE}
&  \quad \quad  \quad \; P_{cp} = \sum_{j=1}^{G} p_{j}^{dl} \leq P 	\tag{\ref{MMF_cp_SE}-C2} 
\\
& \; \quad  \quad \quad \tau_{p}^{dp} \in \{G, \ldots ,T \} \tag{\ref{MMF_cp_SE}-C3}.
\end{align}
Note that the constraints \eqref{poweryek1_SE} and \eqref{powerdo2_SE} are due to the total available power at the BS, but are slightly different. When we use a dedicated pilot per UT, we obtain a dedicated estimate of the channel of each user. Hence in the downlink we can decide on the amount of power we allocate to the UTs on a per UT basis, e.g., $p_{jk}^{dl}$. On the other hand, for co-pilot transmission, the channel estimates of all UTs within a multicasting group are different just by a scalar coefficient. Hence we just can allocate the power on a per group basis, e.g., $p_{j}^{dl}$. It is straightforward to show that for both $\mathcal{P}1$ and $\mathcal{P}2$, the constraints \eqref{poweryek1_SE} and \eqref{powerdo2_SE} should be met with equality. To see this, assume the contrary, e.g., at the optimal solution of $\mathcal{P}2$ we have $P > P_{cp} = \sum_{j=1}^{G} p_{j}^{dl} $. Then one can increase all the $p_{j}^{dl}$ by a factor of $\frac{P}{P_{cp}}$. This increases each UT's SE, hence improves the minimum SE of the system. This contradicts our assumption. Consequently at the optimal solution of $\mathcal{P}2$, $P = P_{cp}$. In the remainder of this section, we find the optimal solutions to $\mathcal{P}1$ and $\mathcal{P}2$ for the six considered scenarios of Fig. \ref{figint}.

To solve $\mathcal{P}1$ and $\mathcal{P}2$, we use a two-step approach. First, we solve them for any arbitrary value of $\tau_{p}^{dp}$ or $\tau_{p}^{cp}$ and determine their optimal solution in closed-form. Second, we find the optimal value of $\tau_{p}^{dp}$ or $\tau_{p}^{cp}$ by searching over the finite discrete set of all the possible values, thanks to the closed-form obtained in the first step. Given an arbitrary $\tau_{p}^{dp}$, as logarithm is a strictly increasing function, $\mathcal{P}1$ can be replaced with a problem $\mathcal{P}^{\prime}1$ as follows
\begin{align}
\label{MMF_dp_SINR}
\mathcal{P}^{\prime}1: \max_{\{p_{jk}^{dl}\}, \{p_{jk}^{u}\}} \min_{\forall j \in \mathcal{G}} & \min_{\forall k \in \mathcal{K}_{j}} \quad  \mathrm{SINR}_{jk}^{\mathrm{dp}}
\\
& s.t. \quad \quad \text{\ref{MMF_dp_SE}-C1 and } P_{dp} = P. \notag
\end{align}
Similarly, $\mathcal{P}2$ can be replaced with a problem $\mathcal{P}^{\prime}2$ as follows
\begin{align}
\label{MMF_cp_SINR}
\mathcal{P}^{\prime}2: \max_{\{p_{j}^{dl}\}, \{p_{jk}^{u}\}} \min_{\forall j \in \mathcal{G}} & \min_{\forall k \in \mathcal{K}_{j}} \quad  \mathrm{SINR}_{jk}^{\mathrm{cp}}
\\
& s.t. \quad \quad \text{\ref{MMF_cp_SE}-C1 and } P_{cp} = P. \notag
\end{align}


\subsection{MMF solution for MRT-undp}
\begin{theorem}
	\label{T-MMF-MRT-undp}
	Consider $\mathcal{P}^{\prime}1$ with MRT-undp, then at the optimal solution all the UTs receive the same SINR and it is equal to  
	\begin{align}
	\label{SINR-MMF-MRT-undp}
	\Gamma =  NP \left(\sum_{i=1}^{G} \sum_{k=1}^{K_{j}} \frac{1+\beta_{ik}P}{\gamma_{ik}^{dp*}} \right)^{-1}		
	\end{align}
	with $\gamma_{ik}^{dp*}  =  \dfrac{\tau_{p}^{dp} p_{ik}^{utot} \beta_{ik}^{2} }{1 + \tau_{p}^{dp} p_{ik}^{utot} \beta_{jt}}$. The optimal uplink training and downlink transmission powers of UT $k$ in group $i$ are
	\begin{align}
	p_{ik}^{u*}  =& \; p_{ik}^{utot} 
	\\
	p_{ik}^{dl*} =& \; 	\dfrac{1 + \beta_{ik} P}{ \gamma_{ik}^{dp*} N} \; \Gamma .
	\end{align}
\end{theorem}
\begin{proof}
	The proof is given in Appendix C.
\end{proof}


\subsection{MMF solution for ZF-undp}
\begin{theorem}
	\label{T-MMF-ZF-undp}
	Consider $\mathcal{P}^{\prime}1$ with ZF-undp, then at the optimal solution all the UTs receive the same SINR and it is equal to  
	\begin{align}
	\label{SINRZFundp}
	\Gamma = \dfrac{(N-K_{tot}) P}{\sum_{i=1}^{G} \sum_{k=1}^{K_{i}} \dfrac{1 + (\beta_{ik} - \gamma_{ik}^{dp*}) P}{\gamma_{ik}^{dp*}}}
	\end{align}
	with $\gamma_{ik}^{dp*}  =  \dfrac{\tau_{p}^{dp} p_{ik}^{utot} \beta_{ik}^{2} }{1 + \tau_{p}^{dp} p_{ik}^{utot} \beta_{ik}}$. The optimal uplink training and downlink transmission powers of UT $k$ in group $i$ are
	\begin{align}
	\label{upTrZFundp}
	p_{ik}^{u*}  =& \; p_{ik}^{utot} 
	\\
	\label{dlTrZFundp}
	p_{ik}^{dl*} =&   \dfrac{1 + (\beta_{ik} - \gamma_{ik}^{dp*}) P}{ \gamma_{ik}^{dp*} ( N-K_{tot})} \; \Gamma .
	\end{align}
\end{theorem}
\textit{Proof Sketch.} The proof is similar to the proof of Theorem \ref{T-MMF-MRT-undp} and its sketch is presented for brevity. First it should be shown that for every UT $k$ in group $i$ its SINR is monotonically increasing with $p_{ik}^{u}$ which results in \eqref{upTrZFundp}. Then it should be shown that at the optimal solution all UTs will have the same SINR, which also determines \eqref{dlTrZFundp}. Now using this fixed value for the SINR and the downlink transmission power constraint, we obtain \eqref{SINRZFundp}. \qed

Remark \ref{Rem-MRTtoZF} described the similarities between the SE expressions with MRT and ZF, and the same pattern appears in the optimal solutions to the MMF problem. As we switch from the MRT to ZF in Theorems \ref{T-MMF-MRT-undp} and \ref{T-MMF-ZF-undp}, the coherent beamforming gain reduces from $N$ to $N-K_{tot}$. Also the interference in the denominator reduces from $\dfrac{ \beta_{ik}  P}{ \gamma_{ik}^{dp*}}$ to $\dfrac{ (\beta_{ik} - \gamma_{ik}^{dp*}) P}{ \gamma_{ik}^{dp*}}$. This is due to the fact that ZF uses the degrees of freedom provided by the large-scale antenna array to cancel the interference toward other UTs at the cost of reducing the desired signal power at each UT. 

\subsection{MMF solution for MRT-mudp}
\begin{corollary}
	\label{C-MMF-MRT-mudp}
	Consider $\mathcal{P}^{\prime}1$ with MRT-mudp, then at the optimal solution all the UTs receive the same SINR and it is equal to \eqref{SINR-MMF-MRT-undp}.
\end{corollary}
\begin{proof}
	From Theorem \ref{T-MRT-mudp}, we know MRT-mudp is equivalent to MRT-undp. Hence it provides the same  SINR for each UT. Therefore its optimal solution is the same as Theorem~\ref{T-MMF-MRT-undp}.
\end{proof}

\subsection{MMF solution for ZF-mudp}
\begin{theorem}
	\label{T-MMF-ZF-mudp}
	Consider $\mathcal{P}^{\prime}1$ with ZF-mudp, then at the optimal solution all the UTs receive the same SINR, i.e., $  \Gamma = \mathrm{SINR}_{ik}^{\mathrm{ZF-mudp*}} \; \forall i,k$, and it is the solution of the equation 
	\begin{align}
	\label{sinr_equation}
	P = \sum_{i=1}^{G}	\dfrac{\Gamma \Delta_{i}}{N-\nu_{i} - \Gamma K_{i}} 
	\end{align}
	where $\Delta_{i} = \sum_{k=1}^{K_{i}} \left( \dfrac{1}{\gamma_{ik}^{dp*}} + P \dfrac{\beta_{ik}}{\gamma_{ik}^{dp*}} - P \right) $ with $\gamma_{ik}^{dp*}  =  \dfrac{\tau_{p}^{dp} p_{ik}^{utot} \beta_{ik}^{2} }{1 + \tau_{p}^{dp} p_{ik}^{utot} \beta_{ik}}$ and $\Gamma < \min_{i \in \mathcal{G}} \{ \frac{N-\nu_{i}}{K_{i}} \}$. Also the optimal uplink training and downlink transmission powers of UT $k$ in group $i$ are
	\begin{align}
	p_{ik}^{u*}  =& \; p_{ik}^{utot} 
	\\
	p_{ik}^{dl*} =& \dfrac{\Gamma}{N-\nu_{i}} \left( \frac{1}{\gamma_{ik}^{dp*}} + P_{i}^{dl} + P \dfrac{\beta_{ik}}{\gamma_{ik}^{dp*}} - P \right) 
	\end{align}
	where $P_{i}^{dl} =  \dfrac{\Gamma \Delta_{i}}{N-\nu_{i} - \Gamma K_{i}} $.
\end{theorem}

\begin{proof}
	The proof is given in Appendix D.
\end{proof}
Note that as the right hand side of \eqref{sinr_equation} is an increasing function of $\Gamma$, its solution can simply be obtained by line search.


\subsection{MMF solution for MRT-mucp}
\begin{theorem}
	\label{T-MMF-MRT-mucp}
	Consider $\mathcal{P}^{\prime}2$ with MRT-mucp, then at the optimal solution all the UTs receive the same SINR and it is equal to  
	\begin{align}
	\label{MRTmucpSINR}
	\Gamma = \dfrac{NP}{\sum_{i=1}^{G} \dfrac{1 + \tau_{p}^{cp} \sum_{m=1}^{K_{i}} p_{im}^{u*} \beta_{im}}{\tau_{p}^{cp} \Upsilon_{i}}}
	\end{align}
	with $\Upsilon_{i} = \min_{t \in \mathcal{K}_{i}} \dfrac{\beta_{it}^{2} p_{it}^{tot}}{1+P \beta_{it}} \; \forall i \in \mathcal{G}$. The optimal uplink training and downlink transmission powers of UT $k$ in group $i$ are
	\begin{align}
	p_{ik}^{u*}  =& \dfrac{1+P \beta_{ik}}{\beta_{ik}^{2}} \; \Upsilon_{i} \quad \quad \forall i \in \mathcal{G}, K \in \mathcal{K}_{i}
	\\
	p_{i}^{dl*} =& \dfrac{\Gamma (1 + \tau_{p}^{cp} \sum_{m=1}^{K_{i}} p_{im}^{u*} \beta_{im})}{\tau_{p}^{cp} N \Upsilon_{i}} \quad \quad \forall j \in \mathcal{G}.
	\end{align}
\end{theorem}

\begin{proof}
	The proof is given in Appendix E.
\end{proof}


\subsection{MMF solution for ZF-mucp}
\begin{theorem}
	\label{T-MMF-ZF-mucp}
	Consider $\mathcal{P}^{\prime}2$ with ZF-mucp, then at the optimal solution all the UTs receive the same SINR and it is equal to 
	\begin{align}
	\label{ZFmucpSINR}
	\Gamma &= \dfrac{P (N-G)}{\sum_{j=1}^{G} \frac{1}{\Delta_{j}}}
	\end{align}
	with $\Delta_{j} = \dfrac{\tau_{p}^{cp} \Upsilon_{j}}{1 + \tau_{p}^{cp}(E_{j} - P \Upsilon_{j})}$, $E_{j} = K_{j} \Upsilon_{j} P + \Upsilon_{j} \sum_{m=1}^{K_{j}} \dfrac{1}{\beta_{jm}}$, and $\Upsilon_{j} = \min_{k \in \mathcal{K}_{j}} \dfrac{p_{jk}^{utot} \beta_{jk}^{2}}{1+\beta_{jk}P} \; \forall j \in \mathcal{G}$. The optimal uplink training and downlink transmission powers of UT $k$ in group $i$ are
	\begin{align}
	p_{ik}^{u*} =&  \dfrac{1+\beta_{ik}P}{\beta_{ik}^{2}} \Upsilon_{i} \quad \quad \forall k \in \mathcal{K}_{i}, \forall i \in \mathcal{G}
	\\
	\label{ZFmucpDLpower}
	p_{i}^{dl*} =&  \left( \sum_{j=1}^{G} \frac{\Delta_{i}}{ \Delta_{j}} \right)^{\!\!\!-1} \!\! P  \quad \forall j \in \mathcal{G} .
	\end{align}
\end{theorem}
\begin{proof}
	The proof is given in Appendix F.
\end{proof}

The achieved results (Theorems \ref{T-MMF-MRT-undp} to \ref{T-MMF-ZF-mucp} and Corollary \ref{C-MMF-MRT-mudp}), determine the optimal value of the SINR, the uplink training powers, and the downlink transmission powers in closed-form, for any given pilot length. These closed-form results enable us to find the optimal value of SE by simply searching over $\tau_{p}^{dp} \in \{K_{tot},\ldots,T\}$ or $\tau_{p}^{dp} \in \{G, \ldots, T\}$ and find the pilot length that provides the highest SE.

\section{Numerical Analysis and Further Discussions}
In this section, we use the results of Section IV to perform a numerical analysis and propose a guideline for multicasting design in massive MIMO systems. In our simulations we consider a system with $G$ multicasting groups where each group has $K$ UTs, i.e., $K_{i} = K \; \forall i \in \mathcal{G}$. The cell radius is considered to be $500$ meters and the UTs are randomly and uniformly distributed in the cell excluding an inner circle of radius $35$ meters. The large-scale fading parameters are modeled as $\beta_{ik} = \bar{d}/ x_{ik}^{\nu}$ where $\nu=3.76$ is the path-loss exponent and the constant $\bar{d} = 10^{-3.53}$ regulates the channel attenuation at $35$ meters \cite{3GPPmodel}. Also $ x_{ik}$ is the distance between UT $k$ in group $i$ and the BS in meters. At a carrier frequency of $2$ GHz, the transmission bandwidth (BW) is assumed to be $20$ MHz, the coherence bandwidth and coherence time are considered to be $300$ kHz and $2.5$ ms, which results in a coherence interval of length $750$ symbols for a vehicular system with speed of $108$ kilometers per hour \cite{marzetta2016fundamentals}. The noise power spectral density is considered to be $-174$ dBm/Hz.

Fig. \ref{Fig2} studies the effect of the system parameters, i.e., $G$, $K$, $N$, $p_{jk}^{utot}$, and $P$, on the optimal SEs that can be obtained for the six scenarios depicted in Fig. \ref{figint}. Figs. \ref{a}, \ref{c}, and \ref{e} represent the high SNR regime, where for the cell-edge, the training SNR is $-5.8$ dB (equivalent to $p_{jk}^{utot}=1$ Watt over the BW) and the downlink SNR is $10$ dB (equivalent to $P=40$ Watt over the BW). Also Figs. \ref{b}, \ref{d}, and \ref{f} are representing the low SNR regime, where for the cell-edge, the training SNR is $-15.8$ dB (equivalent to $p_{jk}^{utot}=0.1$ Watt over the BW) and the downlink SNR is $-5.8 $dB (equivalent to $P=1$ Watt over the BW).

\begin{figure}[]
	\centering
	\begin{subfigure}[b]{0.48\linewidth}
		\centering
		\includegraphics[width=1\columnwidth, trim={3.7cm 8.25cm 4.2cm 9cm},clip]{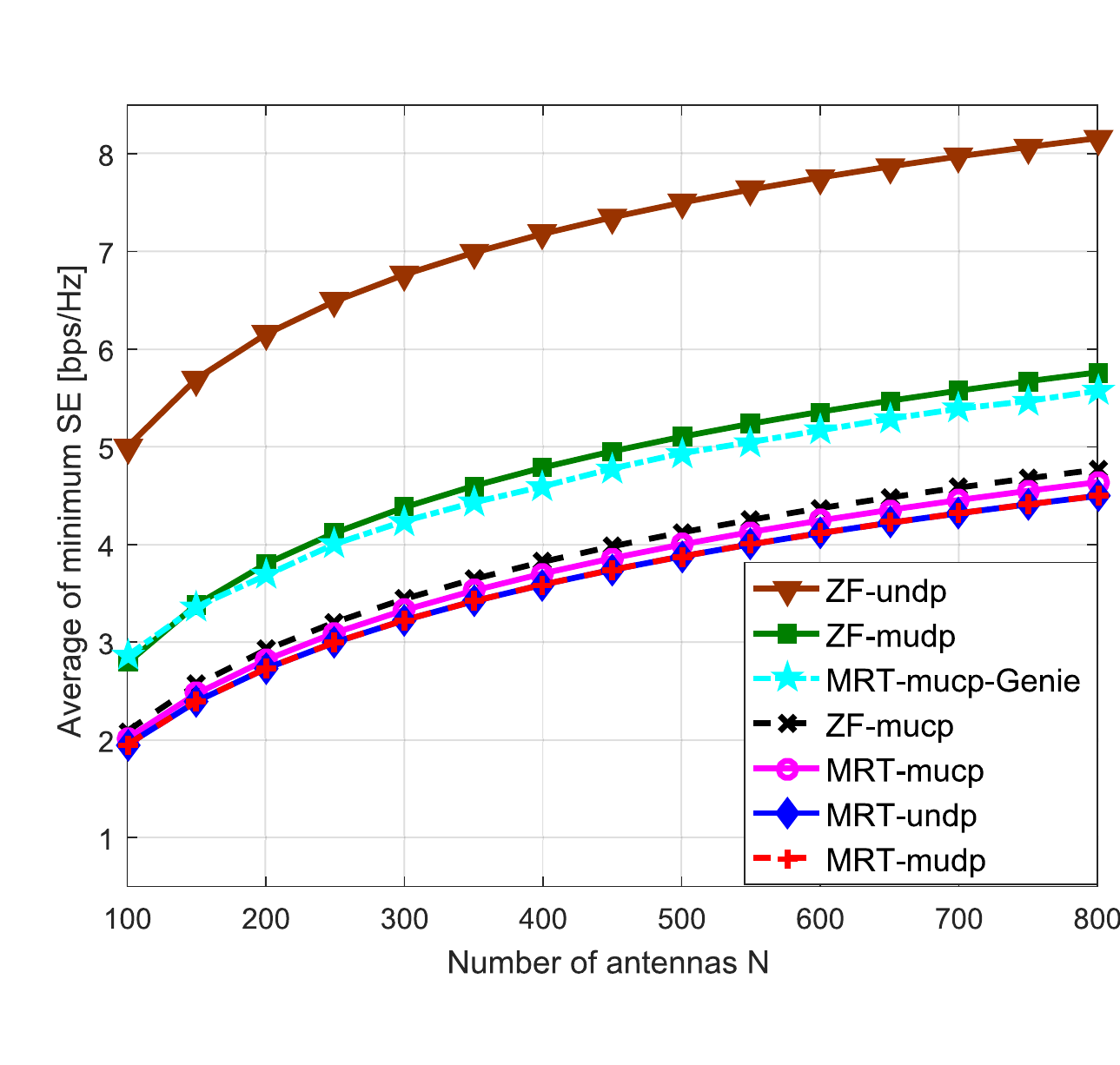}
		\caption{$G\!=\!3$, $K\!=\!10$, $P\!=\!40$, and $p^{utot}\!=\!1$ Watt.}
		\label{a}
	\end{subfigure}%
	~
	\begin{subfigure}[b]{0.48\linewidth}
		\centering
		\includegraphics[width=1\columnwidth, trim={4cm 8.25cm 4cm 9cm},clip]{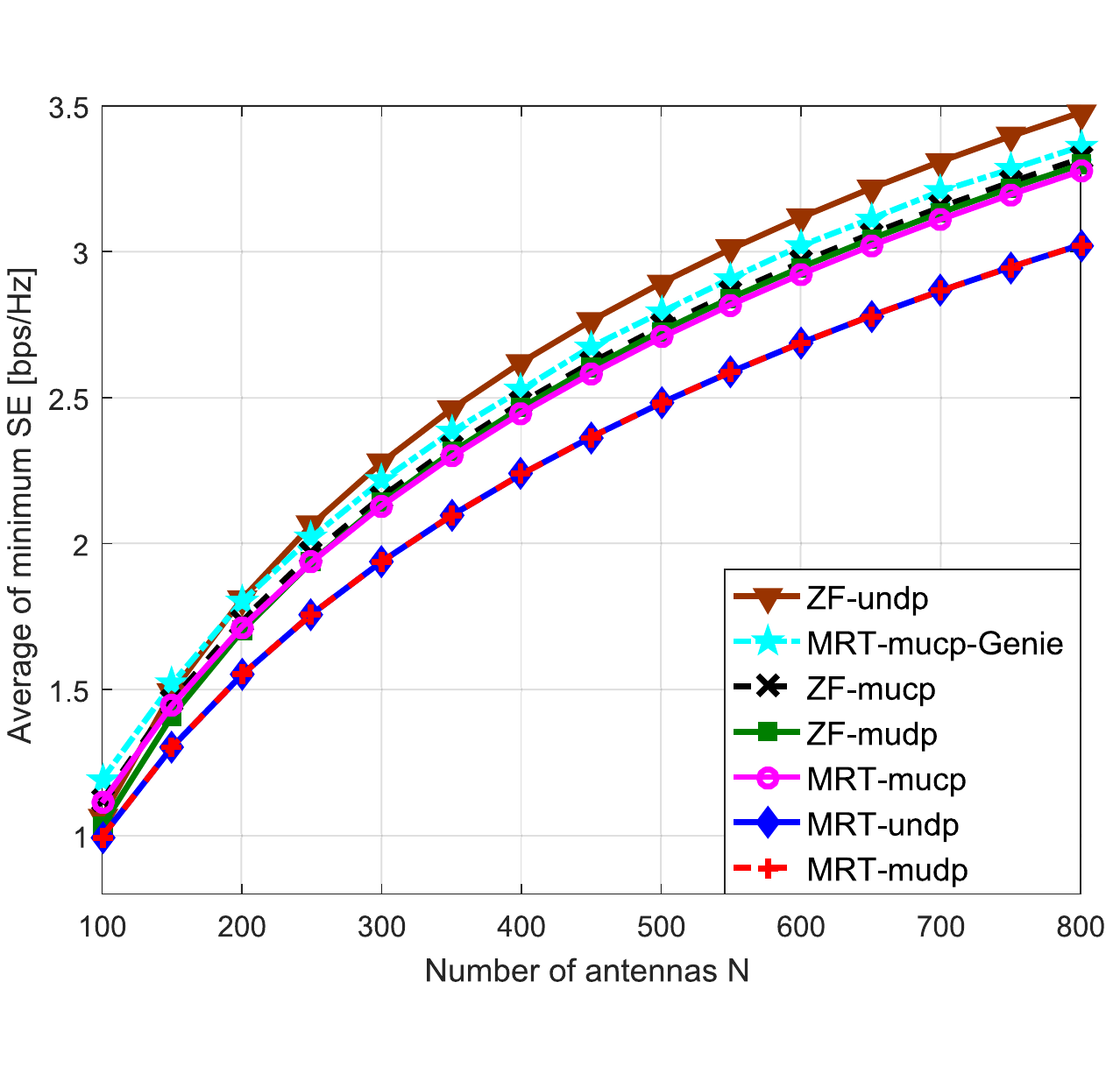}
		\caption{$G\!=\!3$, $K\!=\!10$, $P\!=\!1$, and $p^{utot}\!=\!0.1$ Watt.}
		\label{b}
	\end{subfigure}
	~
	\begin{subfigure}{0.48\linewidth}
		\centering
		\includegraphics[width=1\columnwidth, trim={3.7cm 8.25cm 4.2cm 9cm},clip]{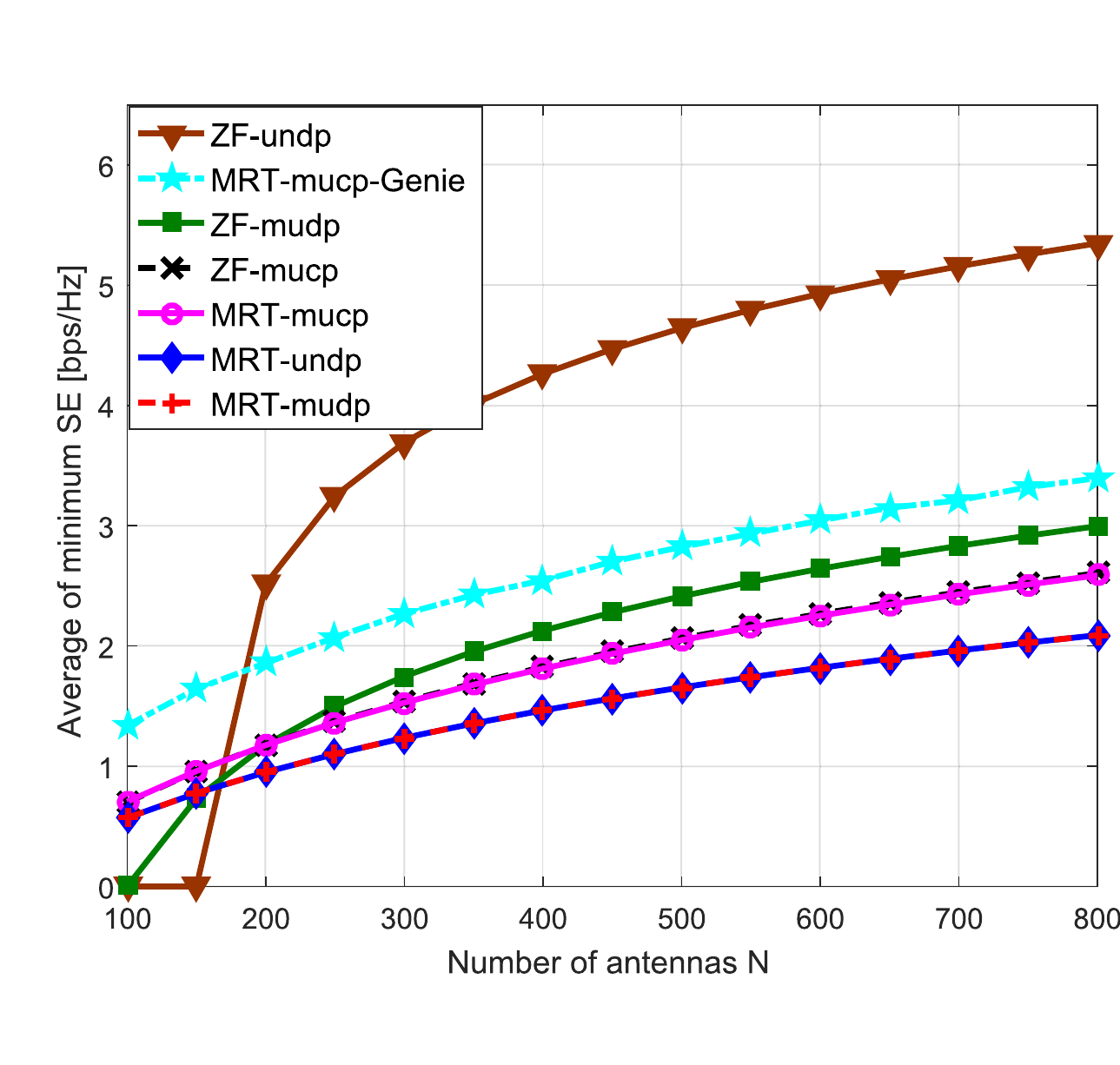}
		\caption{$G\!=\!3$, $K\!=\!50$, $P\!=\!40$, and $p^{utot}\!=\!1$ Watt.}
		\label{c}
	\end{subfigure}
	~
	\begin{subfigure}{0.48\linewidth}
		\centering
		\includegraphics[width=1\columnwidth,trim={4cm 8.25cm 4cm 9cm},clip]{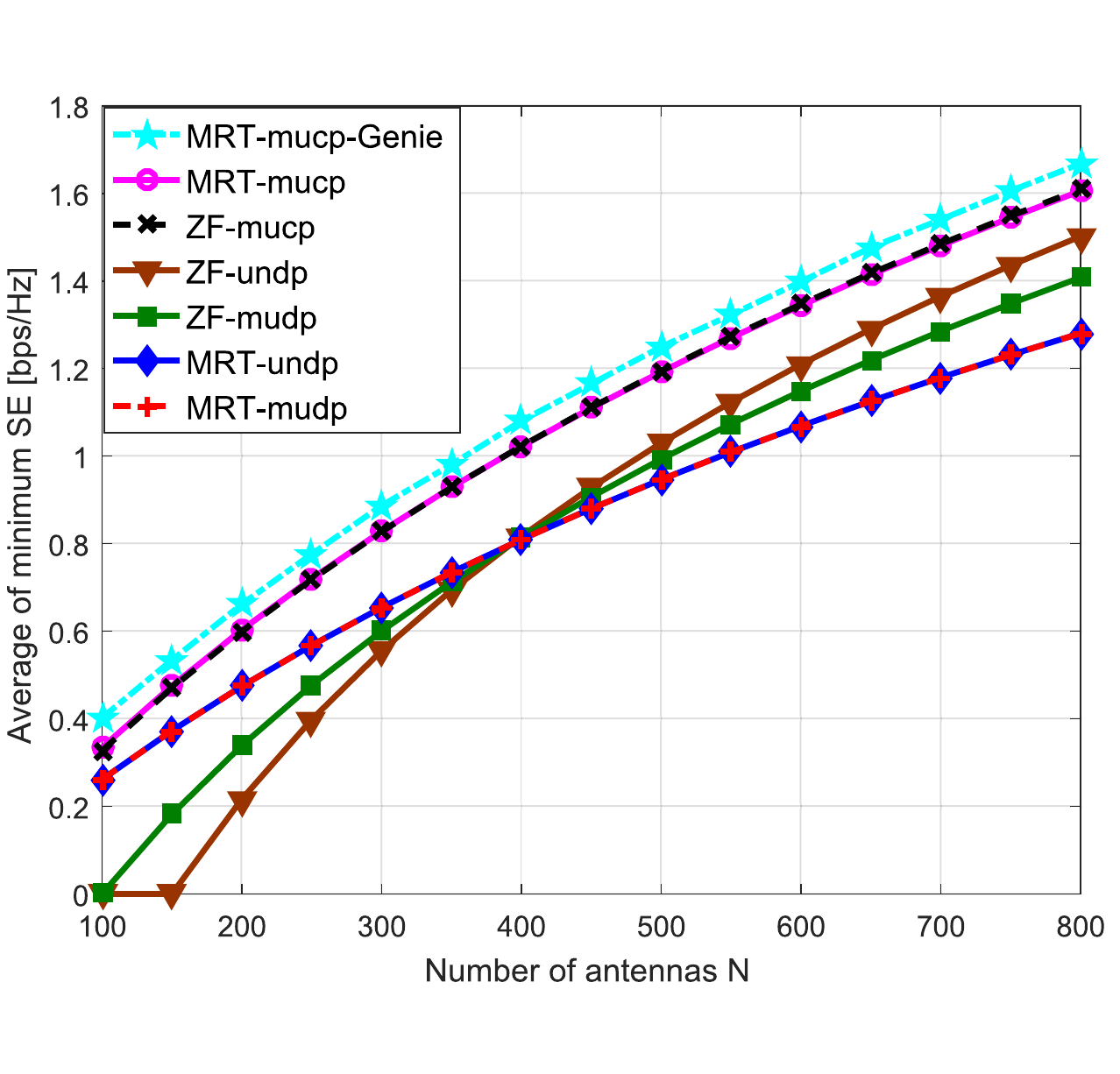}
		\caption{$G\!=\!3$, $K\!=\!50$, $P\!=\!1$, and $p^{utot}\!=\!0.1$ Watt.}
		\label{d}
	\end{subfigure}
	~
	\begin{subfigure}{0.48\linewidth}
		\centering
		\includegraphics[width=1\columnwidth, trim={3.7cm 8.25cm 4.2cm 9cm},clip]{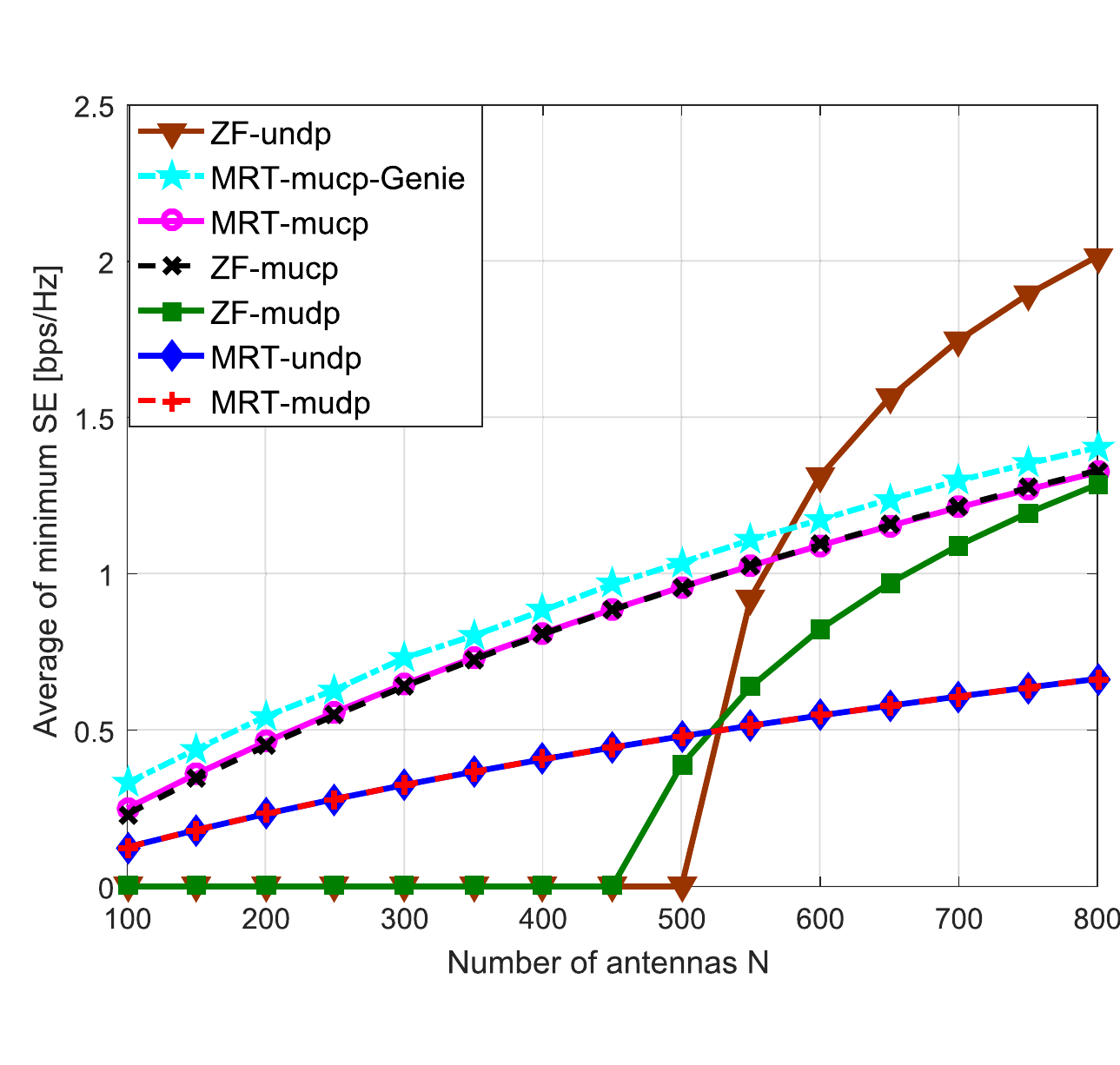}
		\caption{$G\!=\!10$, $K\!=\!50$, $P\!=\!40$, and $p^{utot}\!=\!1$ Watt.}
		\label{e}
	\end{subfigure}
	~
	\begin{subfigure}{0.48\linewidth}
		\centering
		\includegraphics[width=1\columnwidth,trim={4cm 8.25cm 4cm 8.9cm},clip]{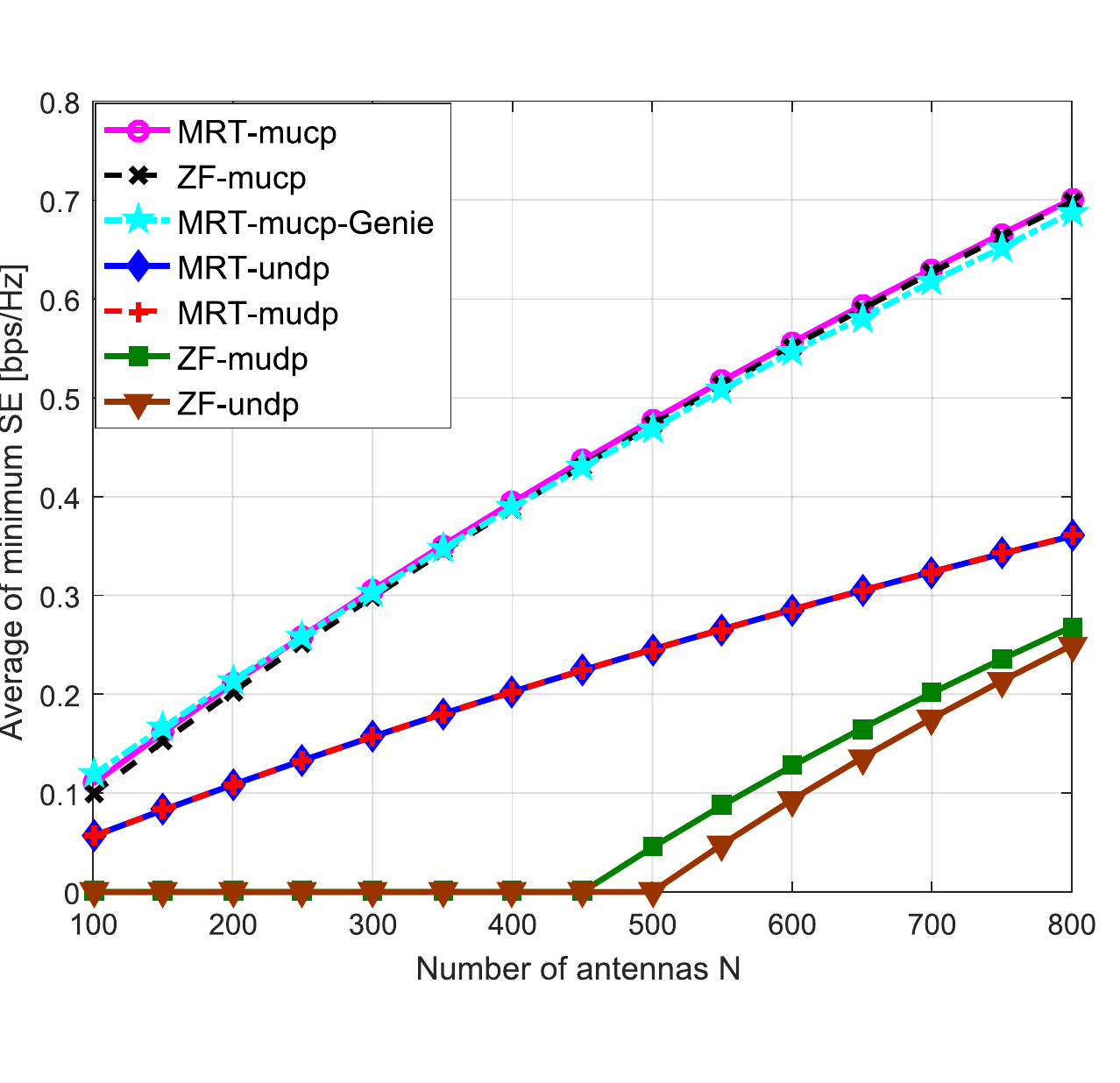}
		\caption{$G\!=\!10$, $K\!=\!50$, $P\!=\!1$, and $p^{utot}\!=\!0.1$ Watt.}
		\label{f}
	\end{subfigure}
	\caption{SE versus $N$ for different system setups.}
	\label{Fig2}
\end{figure}

From Fig. \ref{Fig2} we make the following observations:
\begin{itemize}
	\item The dedicated pilot assignment is more vulnerable to SNR reduction than the co-pilot assignment, comparing the two SNR regimes. For example, consider $N=600$, then the average reduction in SE of ZF-undp comparing Figs. \ref{a}, \ref{c}, \ref{e} respectively with Figs. \ref{b}, \ref{d}, \ref{f} is $6.85$ times while with MRT-mucp and ZF-mucp it is $1.69$ times. This is because the emphasis in dedicated pilot assignment is on achieving good channel estimates, while the co-pilot assignment is focusing on saving time-frequency resources. Hence in the low SNR regime as long as $K_{tot}$ is large enough, e.g. $K_{tot} \gtrapprox 0.2 N$ , MRT-mucp and ZF-mucp provide better performance than other schemes.
	
	\item  In the high SNR regime, ZF-undp significantly outperforms the co-pilot approaches as soon as $N$ becomes slightly bigger than $K_{tot}$ ($N \gtrapprox 1.15 K_{tot}$), as it can be verified from Figs. \ref{a}, \ref{c}, and \ref{e}. The reason is twofold. First, with dedicated pilot assignment a pilot contamination free channel estimation is achieved. Second, in high SNR regime $\tau_{p}^{dp}$ is close to $K_{tot}$, and as $N \gtrapprox 1.15 K_{tot}$ there are enough time-frequency resources for downlink transmission. While for co-pilot assignment the channel estimates are highly contaminated due to the shared pilots.
	
	\item It is plausible that MRT-mucp can provide a better SE than ZF-undp if there is downlink pilot transmission, as downlink pilot transmission can be done efficiently just by employing $G$ symbols of the coherence interval for downlink training \cite{NoDownlinkPilot}. Therefore in Fig. \ref{Fig2} we also have presented the minimum SE of MRT-mucp with genie UTs, i.e., MRT-mucp-Genie, where we assume the UTs \textit{perfectly} estimate their channels from $G$ downlink training symbols. Even in this case, in the high SNR regime, ZF-undp significantly outperforms the MRT-mucp with genie UTs, as soon as $N$ becomes slightly bigger than $K_{tot}$, e.g., $N \gtrapprox 1.2 K_{tot}$, see Figs. \ref{a}, \ref{c}, and \ref{e}.

	\item The SE of the co-pilot assignment approaches is more robust to adding more UTs to the system than the SE of the dedicated pilot assignment approaches. For example, consider $N=700$ and compare the SE of ZF-undp and MRT-mucp in Figs. \ref{c} and \ref{d} (where $K_{tot} = 150$) respectively with Figs. \ref{e} and \ref{f} (where $K_{tot} = 500$). For ZF-undp the SE reduces by a factor of $2.95$ (comparing Fig. \ref{c} with Fig. \ref{e}) and $7.77$ (comparing Fig. \ref{d} with Fig. \ref{f}) while for MRT-mucp it reduces by a factor of $2$ (comparing Fig. \ref{c} with Fig. \ref{e}) and $2.36$ (comparing Fig. \ref{d} with Fig. \ref{f}). This is because adding more UTs increases the pilot overhead in dedicated pilot assignment approaches while it has a slight effect for co-pilot approaches. Hence co-pilot approaches are more suitable for applications like DVB-H or mobile TV over wide areas with many users \cite{DVB-GFaria,DVB-Elhajjar}.
	
	\item As we increase $K_{tot}$ by adding more multicasting groups, e.g., in applications with large number of multicasting UTs such as DVB-H \cite{DVB-GFaria}, the downlink training becomes less important and can be neglected, e.g., compare Figs \ref{a}, \ref{c}, and \ref{e} or Figs \ref{b}, \ref{d}, and \ref{f}. This is because adding more groups requires more time-frequency resources for downlink training.
	
	\item MRT-mucp nearly provides the same SE as ZF-mucp, e.g. see Figs. \ref{c}, \ref{d}, \ref{e}. This is because the deliberate pilot contamination that was enforced to the precoder structure, \eqref{ZFMUCP}, prevents the ZF-based pecoder from suppressing the interference efficiently. Therefore due to the higher complexity of ZF, if the co-pilot strategy is employed, it is beneficial to use MRT-mucp rather than ZF-mucp.
	
	\item MRT-mucp always outperform MRT-undp and MRT-mudp, e.g., see Figs \ref{e} and \ref{f}. Hence if MRT is employed for multicasting, it is better to use the MRT-mucp scheme.
	
	\item In all of the considered setups in Fig. \ref{Fig2}, the maximum performance is either achieved by ZF-undp or MRT-mucp. Hence a multicasting system need to support these two transmission modes and switch between them depending on the system parameters.
	
	\item As detailed in Remark \ref{remZFmudp}, ZF-mudp is the generalized version of the precoder proposed in \cite{MeysamMultiComplexity} and it outperforms the SDR-based precoding schemes \cite{karipidis2008quality}. Also ZF-mudp is always outperformed by either MRT-mucp or ZF-undp. Therefore in a massive MIMO system that accounts for CSI acquisition, a system with hybrid transmission that switches between MRT-mucp and ZF-undp outperforms SDR-based approaches \cite{karipidis2008quality,MeysamMultiComplexity}. 
\end{itemize}

The aforementioned observations were achieved either at high or low SNR regime. Fig. \ref{SNR} verifies them for a wide range of SNR. Considering $N=300$, $G=4$, $K=50$, Fig \ref{SNRa} presents the SE of the proposed scheme for a fixed cell edge training SNR of $-5.8$ dB, while the cell edge downlink SNR is changing from $-20$ dB to $20$ dB. Fig. \ref{SNRb} presents the SE for a fixed cell edge downlink SNR of $10$ dB while the cell edge training SNR is changing from $-30$ dB to $5$ dB. Note that the same observation holds true, e.g., 1) MRT-mucp and ZF-mucp have the same performance; 2) The optimal performance is achieved by switching between MRT-mucp and ZF-undp; 3) at low SNR the co-pilot approaches perform better than the dedicated approaches, and the opposite holds for high SNR; and 4)  MRT-mucp always outperform MRT-undp and MRT-mudp. 

\begin{figure}[]
	\centering
	\begin{subfigure}[b]{0.48\linewidth}
		\centering
		\includegraphics[width=1\columnwidth, trim={4.1cm 8.4cm 4.4cm 9cm},clip]{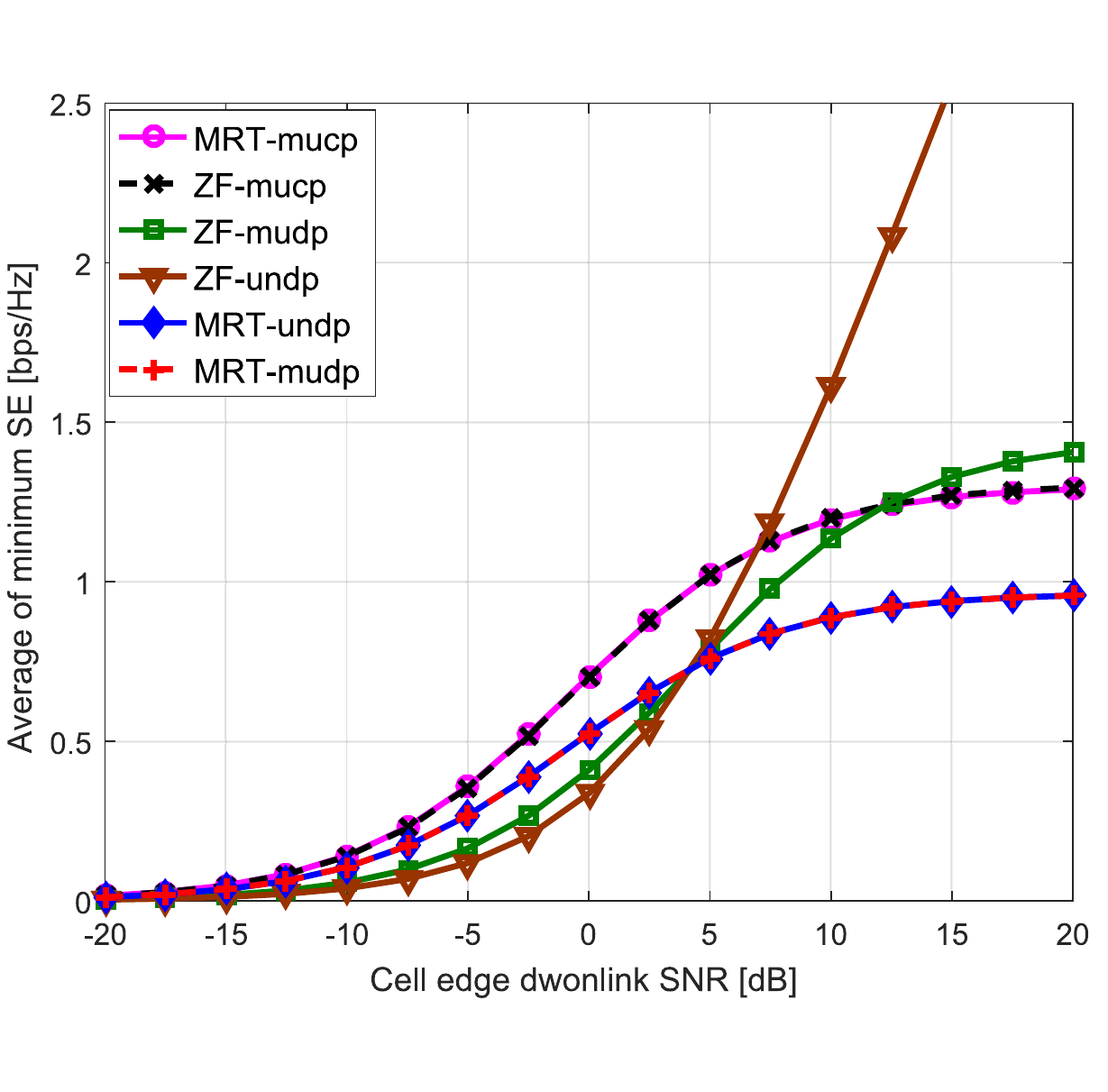}
		\caption{SE vs Downlink SNR.}
		\label{SNRa}
	\end{subfigure}%
	~
	\begin{subfigure}[b]{0.48\linewidth}
		\centering
		\includegraphics[width=1\columnwidth, trim={4cm 8.4cm 4.4cm 9.1cm},clip]{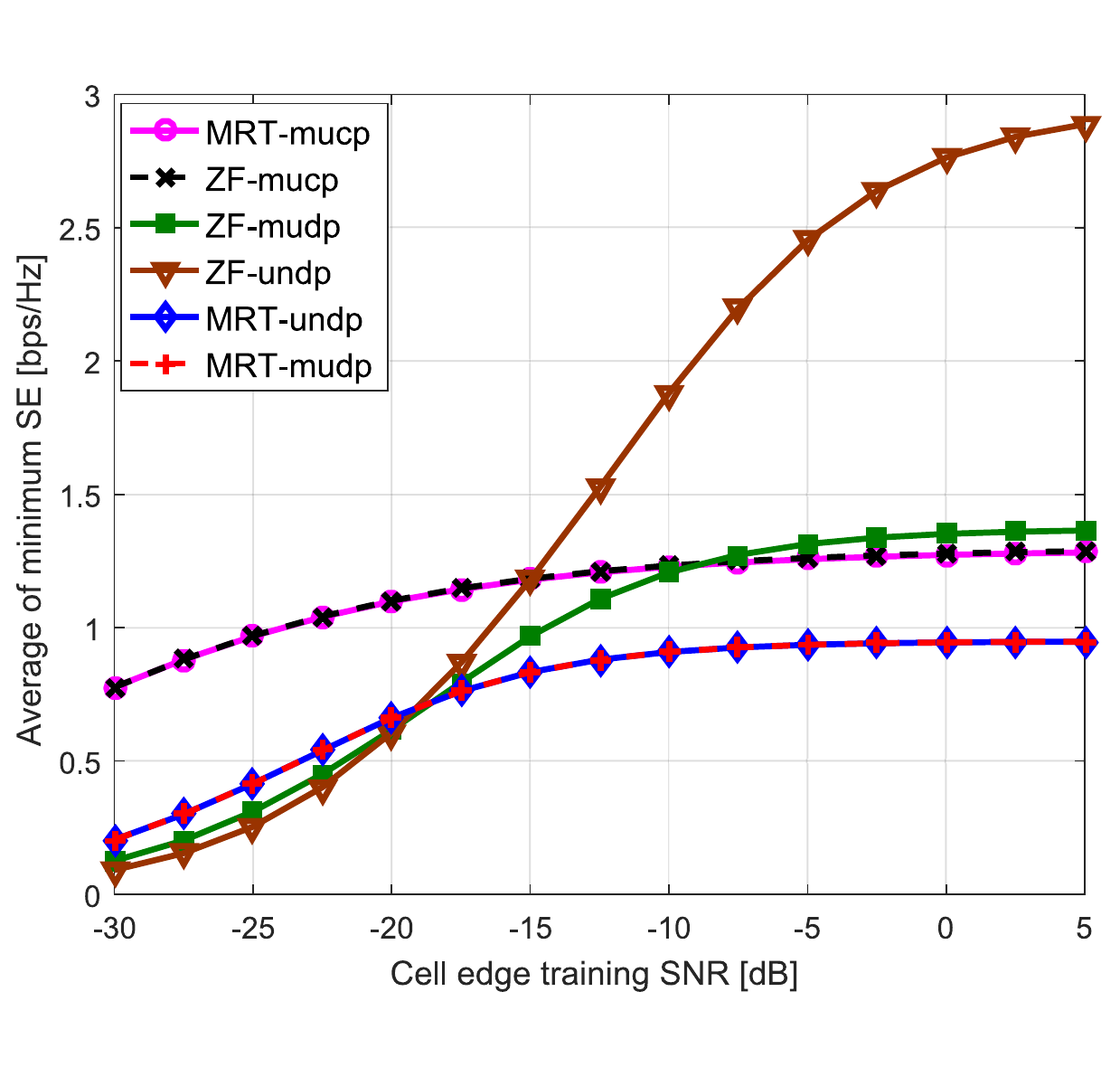}
		\caption{SE vs Training SNR.}
		\label{SNRb}
	\end{subfigure}
	\caption{SE versus SNR.}
	\label{SNR}
\end{figure}

As some of the state of the art multicast standards and applications, e.g. DVB-H and mobile TV, employ omnicast transmission \cite{DVB-GFaria,DVB-Elhajjar}, it is interesting to compare the performance of the proposed multicast schemes with omnicast transmission. Therefore in Fig. \ref{figOmni} we consider a system with $P=40$ Watt, $p^{uot}_{jk}=1$ Watt, and $G$ multicasting groups where $G$ is changing from $1$ to $30$ with $K$ UTs per group. It presents the minimum SE versus the number of multicasting groups for the proposed multicasting schemes and the omnicast transmission. For omnicast transmission we assume the channels are perfectly known at the UTs, and the minimum SE is computed as follows
\begin{align}
\label{OmniEq}
\mathrm{SE}_{\mathrm{Omnicast}} = \mathbb{E} \left[ \min_{\{ j \},\{ k \}} \mathbb{E} \left[\dfrac{1}{G} \log_{2} \left(1 + \frac{P \Vert \mathbf{h}_{jk} \Vert^{2} }{\sigma^{2}}\right) \rvert \beta_{jk} \right]  \right]
\end{align}
where the outer expectation is with respect to large-scale fading and the inner expectation is with respect to small-scale fading. Note that \eqref{OmniEq} provides an upper bound on the performance of an omnicast transmission as we assumed perfect channel knowledge at the UTs. In practice, terminals will have to rely on channel estimates obtained from downlink pilots. This pilot transmission is complicated by the fact that optimal training entails the transmission of mutually orthogonal pilots from each antenna; with a large number of antennas, this pilot overhead can be significant. A reduction of the pilot overhead, at the cost of some spatial diversity order loss, can be achieved by transmission into a pre-determined subspace \cite{meng2016omnidirectional,karlsson2014operation}.  Note that in independent Rayleigh fading, a conventional omnicast system that uses a single antenna is equivalent to the considered array \cite{karlsson2014operation}, while maximal dimensionality reduction applied. A corresponding achievable SE can be obtained from \cite{larsson2016joint}, by setting $\rho_b'=0$ in equation (49) therein\footnote{There is an $M^{\prime}$ parameter in equation (49) of \cite{larsson2016joint}, that in Fig. \ref{figOmni} we found its optimal value by exhaustive search, which gives us the best lower bound that can be obtained for omnicast transmission based on \cite{larsson2016joint}.}, which we refer to as omnicast with imperfect downlink CSI. From Fig. \ref{figOmni} one can see that for any $K_{tot}= GK$, at least ZF-undp or MRT-mucp provide significantly better performance than omnicast transmission. Note that even when we have $K_{tot}=1500$ UTs in the system, MRT-undp provides more than $3$ times higher SE than omnicast transmission. This highly motivates the application of massive MIMO in new multicasting standards \cite{DVB-GFaria,DVB-Elhajjar}.
\begin{figure}[]
	\centering
	\begin{subfigure}[b]{0.48\linewidth}
		\centering
		\includegraphics[width=1\columnwidth, trim={4.1cm 8.3cm 4.4cm 8.85cm},clip]{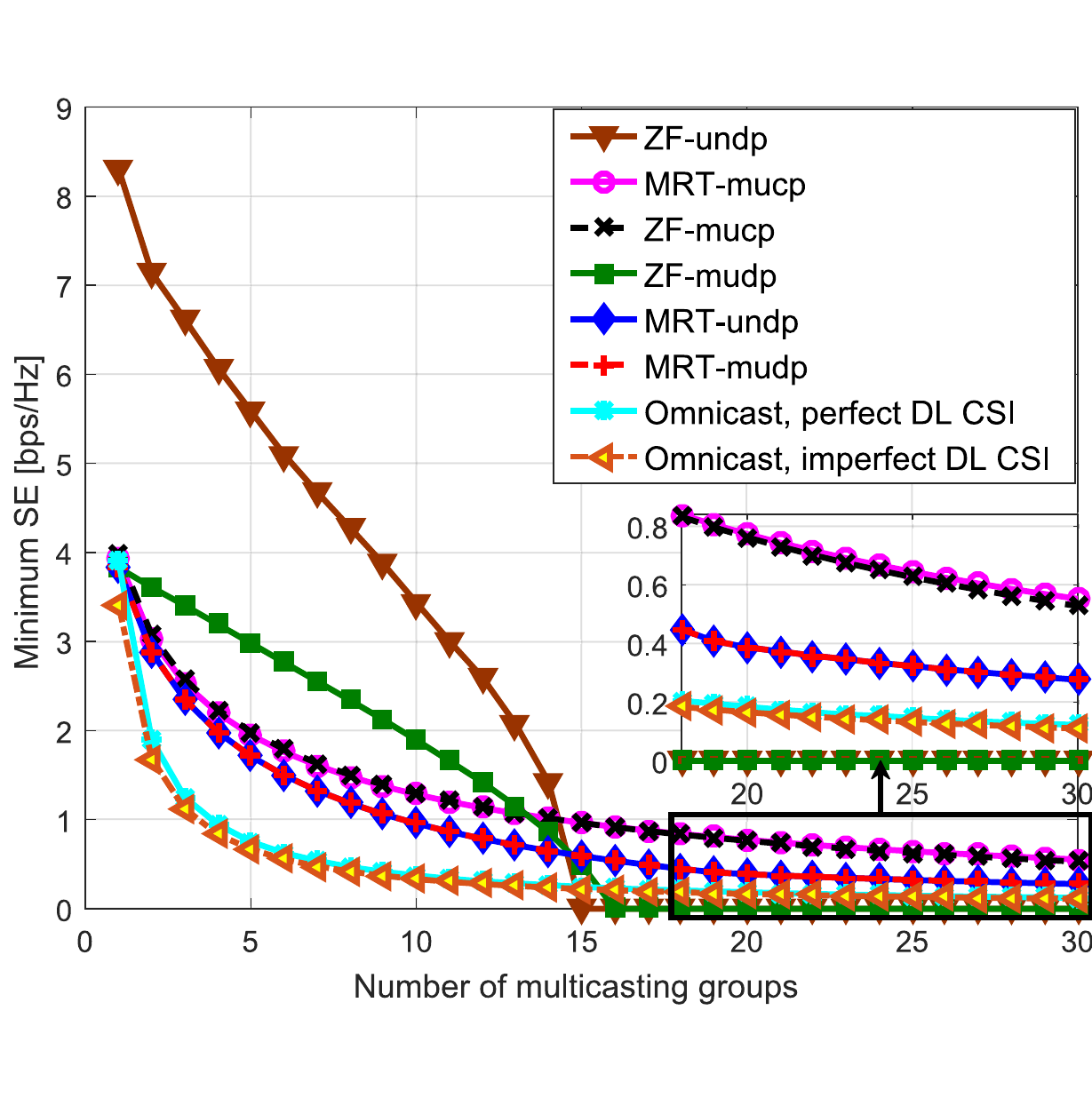}
		\caption{$K=20$ and $N=300$.}
		\label{Omnia}
	\end{subfigure}
	~
	\begin{subfigure}[b]{0.48\linewidth}
		\centering
		\includegraphics[width=1\columnwidth, trim={4cm 8.3cm 4.4cm 8.9cm},clip]{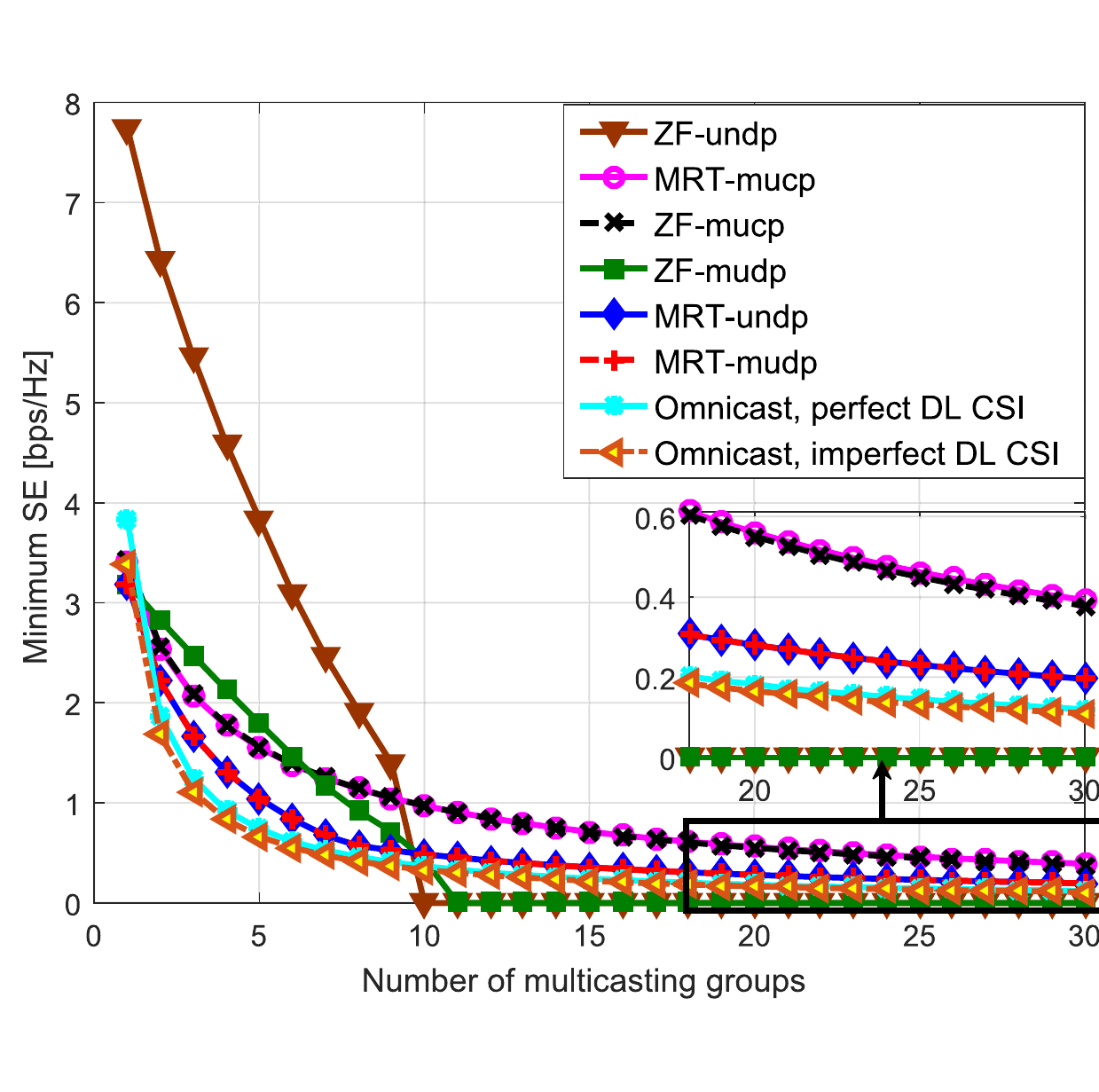}
		\caption{$K=50$ and $N=500$.}
		\label{Omnib}
	\end{subfigure}
	\caption{Comparison between Multicast and Omnicast transmissions.}
	\label{figOmni}
\end{figure}

Based on the numerical analysis provided in this section, Fig. \ref{MuRegimes} presents a guideline for multicasting in massive MIMO systems. Given the system parameters, it determines which scheme should be applied in different scenarios. Also based on our derived results in Section IV, we can explicitly specify the SE that can be obtained using this selected scheme.

\begin{figure}[]
	\centering
	\includegraphics[width=1\columnwidth, trim={0.1cm 5.5cm 3cm 6cm},clip]{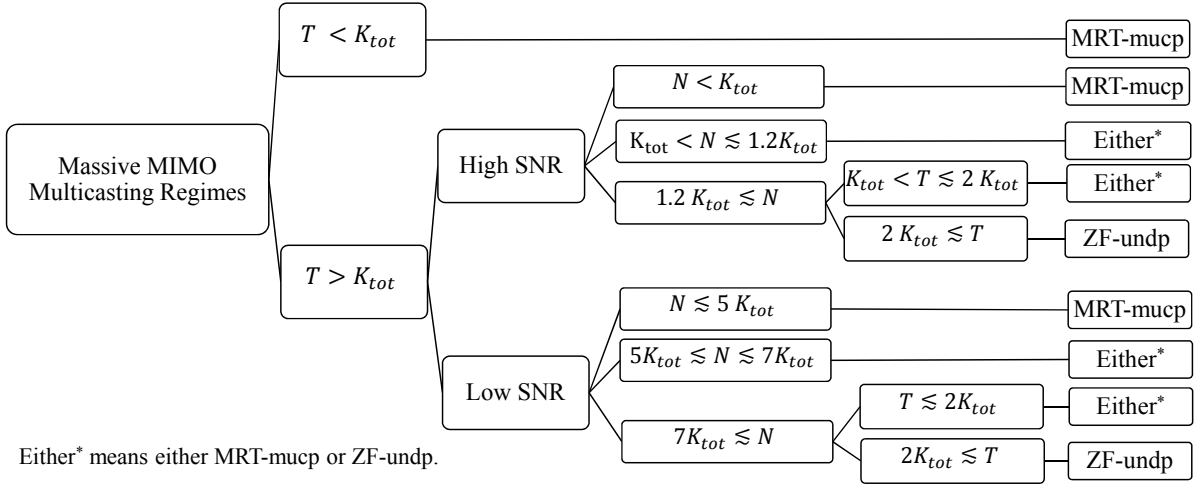}
	\caption{The massive MIMO multicasting regimes.}
	\label{MuRegimes}
\end{figure}

\section{Summary and Conclusion}
In this paper, we studied multi-group multicasting in the context of massive MIMO. First, we introduced different transmission technologies (multicast and unicast), different pilot assignment strategies (co-pilot or dedicated pilot assignment), and the two common precoding schemes in massive MIMO (MRT and ZF). The six possible combinations were outlined in Fig. \ref{figint}. Second, for each of these schemes we derived an achievable SE while accounting for the uplink pilot-based CSI acquisition. Third, for any given training length, we solved the max-min fairness problem for the proposed schemes and found the optimal uplink pilot powers, downlink precoding powers, and the optimal SEs, all in closed-forms. Fourth, based on the achieved results we evaluated the proposed schemes numerically and drew a guideline for practical multi-group massive MIMO multicasting design. We showed that a massive MIMO multicasting system need to support two transmission modes, i.e., MRT-mucp and ZF-undp, and switches between them depending on the system parameters.

\section*{Appendices}
The appendix provides the proof of proposed theorems and propositions. We will frequently use the following lemma, which can be proved by standard techniques (for example see Section II of \cite{marzetta2016fundamentals}).
\begin{lemma}
	\label{MainLemma}
	Consider a discrete memoryless channel with input $x \! \in \! \mathbb{C}$ and output $y\!=\!\!~h x +~ \! v +~\! n$, where $h$ is a deterministic channel coefficient, $v$ is a random interference with zero mean and power $\mathbb{E}[\vert v \vert^{2}] = p_{v}$ that is uncorrelated with $x$, and $n \sim \mathcal{CN}(0,\sigma^{2})$ is independent circularly symmetric complex Gaussian noise. Then if the input power is limited as $\mathbb{E}[\vert x \vert^{2}] = P$ and the channel response $h \in \mathbb{C}$ and interference power $p_{v} \in \mathbb{R}_{+}$ are known at the output, then $\mathrm{SINR} = \dfrac{P \vert h \vert^{2}}{p_{v} + \sigma^{2}}$ and $r = \log_{2} (1 + \mathrm{SINR})$ are the achievable SINR and SE for this channel.
\end{lemma}

\section*{Appendix A - Achievable SE with ZF-mudp}
Starting from \eqref{multicasttransmission} and applying \eqref{ZFMUDP} we have
\begin{align}
y_{ik} \!\! = \!  \underbrace{\mathbb{E}[\hat{\mathbf{g}}_{ik}^{dpH}  \mathbf{w}_{i}^{\mathrm{ZF\!-\!mudp}}]}_{h}  \underbrace{s_{i}}_{x} \! + \! \underbrace{( \hat{\mathbf{g}}_{ik}^{dpH}  \mathbf{w}_{i}^{\mathrm{ZF\!-\!mudp}} \!-\! \mathbb{E}[\hat{\mathbf{g}}_{ik}^{dpH}  \mathbf{w}_{i}^{\mathrm{ZF\!-\!mudp}}]) s_{i} \!-\!  \tilde{\mathbf{g}}_{ik}^{dpH} \! \sum_{j=1}^{G} \! \mathbf{w}_{j}^{\mathrm{ZF\!-\!mudp}} s_{j}}_{v} \! +  n. \label{hANDv}
\end{align}
Now using Lemma \ref{MainLemma} while considering $h$, $x$ and $v$ as shown in \eqref{hANDv}, we obtain the following effective SINR for UT $k$ in group $i$: 
\begin{align}
\label{SINRZFmudpaid}
\mathrm{SINR}_{ik}^{\mathrm{ZF-mudp}}   = \dfrac{\vert \mathbb{E}[\hat{\mathbf{g}}_{ik}^{dpH}  \mathbf{w}_{i}^{\mathrm{ZF-mudp}}]  \vert^{2}  }{1 + \mathrm{var}(\hat{\mathbf{g}}_{ik}^{dpH}  \mathbf{w}_{i}^{\mathrm{ZF-mudp}}) + \sum_{j=1}^{G} \mathbb{E}[ |\tilde{\mathbf{g}}_{ik}^{dpH} \mathbf{w}_{j}^{\mathrm{ZF-mudp}} |^{2} ] }.
\end{align}
Next we find the exact value of each term in \eqref{SINRZFmudpaid}. For the term $\mathbb{E}[\hat{\mathbf{g}}_{ik}^{dpH}  \mathbf{w}_{i}^{\mathrm{ZF-mudp}}]$ we have
\begin{align}
\label{powerterm}
&\mathbb{E}[\hat{\mathbf{g}}_{ik}^{dpH}  \mathbf{w}_{i}^{\mathrm{ZF-mudp}}]
=\sum_{m=1}^{K_{i}} \mathbb{E} \left[ \mathrm{tr} \left(  \sqrt{\mu_{im}} \hat{\mathbf{g}}_{im}^{dp} \hat{\mathbf{g}}_{ik}^{dpH}  (\mathbf{I}_{N} - \hat{\mathbf{G}}_{-i} (\hat{\mathbf{G}}_{-i}^{H} \hat{\mathbf{G}}_{-i})^{-1} \hat{\mathbf{G}}_{-i}^{H} ) \right) \right]
\\
&=\sum_{m=1}^{K_{i}} \mathrm{tr} \left( \mathbb{E} [   \sqrt{\mu_{im}} \hat{\mathbf{g}}_{im}^{dp} \hat{\mathbf{g}}_{ik}^{dpH} ] \mathbb{E} [(\mathbf{I}_{N} - \hat{\mathbf{G}}_{-i} (\hat{\mathbf{G}}_{-i}^{H} \hat{\mathbf{G}}_{-i})^{-1} \hat{\mathbf{G}}_{-i}^{H} )  ] \right) \notag
\\
&=  \sqrt{\mu_{ik}} \gamma_{ik}^{dp} \left( N -  \mathbb{E} \left[ \mathrm{tr} \left(  \hat{\mathbf{G}}_{-i} (\hat{\mathbf{G}}_{-i}^{H} \hat{\mathbf{G}}_{-i})^{-1} \hat{\mathbf{G}}_{-i}^{H}  \right) \right]  \right) 
= \sqrt{\mu_{ik}} \gamma_{ik}^{dp} (N - \nu_{i}). \notag
\end{align}
Now let us consider the interference term due to imperfect CSI. We have
\begin{align}
\label{interferenceZFmudp}
\mathbb{E}[ |\tilde{\mathbf{g}}_{ik}^{dpH} \mathbf{w}_{j}^{\mathrm{ZF-mudp}} |^{2} ] =&  \mathbb{E}[ \tilde{\mathbf{g}}_{ik}^{dpH} \mathbf{w}_{j}^{\mathrm{ZF-mudp}} \mathbf{w}_{j}^{\mathrm{ZF-mudp}H} \tilde{\mathbf{g}}_{ik}^{dp} ] 
\\
=& (\beta_{ik} - \gamma_{ik}^{dp}) \mathrm{tr} ( \mathbb{E}[\mathbf{w}_{j}^{\mathrm{ZF-mudp}} \mathbf{w}_{j}^{\mathrm{ZF-mudp}H}] ) = (\beta_{ik} - \gamma_{ik}^{dp}) \sum_{t=1}^{K_j} p_{jt}^{dl} . \notag
\end{align}
Now we need to calculate the variance term,
\begin{align}
\label{varZFmudp}
\mathrm{var}(\hat{\mathbf{g}}_{ik}^{dpH}  \mathbf{w}_{i}^{\mathrm{ZF-mudp}}) = \mathbb{E}[\vert \hat{\mathbf{g}}_{ik}^{dpH}  \mathbf{w}_{i}^{\mathrm{ZF-mudp}} \vert^{2}] - \vert \mathbb{E}[ \hat{\mathbf{g}}_{ik}^{dpH}  \mathbf{w}_{i}^{\mathrm{ZF-mudp}} ] \vert^{2} .
\end{align}
Denote $ \mathbf{C}_{i} = \mathbf{I}_{N} - \hat{\mathbf{G}}_{-i} (\hat{\mathbf{G}}_{-i}^{H} \hat{\mathbf{G}}_{-i})^{-1} \hat{\mathbf{G}}_{-i}^{H}  $. For the term $\mathbb{E}[\vert \hat{\mathbf{g}}_{ik}^{dpH}  \mathbf{w}_{i}^{\mathrm{ZF-mudp}} \vert^{2}]$ we have  
\begin{align*}\allowdisplaybreaks
&\mathbb{E}[\vert \hat{\mathbf{g}}_{ik}^{dpH}  \mathbf{w}_{i}^{\mathrm{ZF-mudp}} \vert^{2}] = \mathbb{E}[ \hat{\mathbf{g}}_{ik}^{dpH} \mathbf{C}_{i} \sum_{m=1}^{K_{i}} \sqrt{\mu_{im}} \hat{\mathbf{g}}_{im}^{dp} \sum_{t=1}^{K_{i}} \sqrt{\mu_{it}} \hat{\mathbf{g}}_{it}^{dpH} \mathbf{C}_{i} \hat{\mathbf{g}}_{ik}^{dp} ]
\\
&= \underbrace{\mathrm{tr} \left( \mathbb{E} \left[  \hat{\mathbf{g}}_{ik}^{dp} \hat{\mathbf{g}}_{ik}^{dpH} \mathbf{C}_{i} \! \left( \! \sum_{m=1, m\neq k}^{K_{i}} \sum_{t=1, t\neq k}^{K_{i}} \! \! \! \sqrt{\mu_{im} \mu_{it}} \hat{\mathbf{g}}_{im}^{dp} \hat{\mathbf{g}}_{it}^{dpH}  \!\! \right) \!\! \mathbf{C}_{i} \right] \right)}_{(i)}  
+
\underbrace{\mu_{ik} \mathbb{E} \left[ \left( \hat{\mathbf{g}}_{ik}^{dpH} \mathbf{C}_{i} \hat{\mathbf{g}}_{ik}^{dp} \right)^{2}  \right]}_{(ii)}  
\\
&+ \underbrace{\mathrm{tr} \left( \mathbb{E} \left[  \hat{\mathbf{g}}_{ik}^{dp} \hat{\mathbf{g}}_{ik}^{dpH} \mathbf{C}_{i}  \left( \sum_{t=1, t\neq k}^{K_{i}}  \sqrt{\mu_{ik} \mu_{it}} \hat{\mathbf{g}}_{ik}^{dp} \hat{\mathbf{g}}_{it}^{dpH} + \sum_{m=1, m\neq k}^{K_{i}} \sqrt{\mu_{im} \mu_{ik}} \hat{\mathbf{g}}_{im}^{dp} \hat{\mathbf{g}}_{ik}^{dpH}  \right) \mathbf{C}_{i} \right] \right)}_{(iii)}.
\end{align*} 
Notice that $(iii)$ is equal to zero due to the independency of $ \hat{\mathbf{g}}_{ik}^{dp} $ and $ \hat{\mathbf{g}}_{it}^{dp} \; \forall t \neq k , t \in \mathcal{K}_{i}$. The term $(i)$ reduces to
\begin{align*}
&\sum_{m=1, m\neq k}^{K_{i}} \mu_{im}  \mathrm{tr} \left( \mathbb{E} \left[  \hat{\mathbf{g}}_{ik}^{dp} \hat{\mathbf{g}}_{ik}^{dpH} \mathbf{C}_{i}   \hat{\mathbf{g}}_{im}^{dp} \hat{\mathbf{g}}_{im}^{dpH} \mathbf{C}_{i} \right] \right)
=\gamma_{ik}^{dp} \sum_{m=1, m\neq k}^{K_{i}} \mu_{im}  \mathrm{tr} \left( \mathbb{E} \left[ \mathbf{C}_{i}   \hat{\mathbf{g}}_{im}^{dp} \hat{\mathbf{g}}_{im}^{dpH} \mathbf{C}_{i} \right] \right)
\\
&= \gamma_{ik}^{dp} \sum_{m=1, m\neq k}^{K_{i}} \mu_{im} \gamma_{im}^{dp} \mathbb{E} \left[  \mathrm{tr} \left(  \mathbf{C}_{i} \right) \right] \stackrel{(a)}{=}  \gamma_{ik}^{dp} \sum_{m=1, m\neq k}^{K_{i}} p_{im}^{dl}  \notag
\end{align*} 
where in $(a)$ we used the fact that $N - \nu_{i} = \mathrm{tr} \left(  \mathbf{C}_{i} \right) $. For the term $(ii)$, denote $\hat{\mathbf{g}}_{ik}^{dp} = \sqrt{\gamma_{ik}^{dp}} \; \hat{\mathbf{h}}_{ik}$ with $\hat{\mathbf{h}}_{ik} \sim \mathcal{CN}(\mathbf{0},\mathbf{I}_{N})$, then we have
\begin{align}
\mu_{ik} \mathbb{E} \left[ \left( \hat{\mathbf{g}}_{ik}^{dpH} \mathbf{C}_{i} \hat{\mathbf{g}}_{ik}^{dp} \right)^{2}  \right] &= \mu_{ik} (\gamma_{ik}^{dp})^{2} \mathbb{E} \left[ \left( \hat{\mathbf{h}}_{ik}^{H} \mathbf{C}_{i} \hat{\mathbf{h}}_{ik} \right)^{2}  \right] = \dfrac{p_{ik}^{dl} \gamma_{ik}^{dp}}{N - \nu_{i}} \left(  \mathrm{tr}(\mathbf{C}_{i})^{2} + \mathrm{tr}(\mathbf{C}_{i}^{2}) \right)
\\
&=  p_{ik}^{dl} \gamma_{ik}^{dp}   (N -\nu_{i}) +  p_{ik}^{dl} \gamma_{ik}^{dp}.  \notag
\end{align}
Therefore $\mathrm{var}(\hat{\mathbf{g}}_{ik}^{dpH}  \mathbf{w}_{i}^{\mathrm{ZF-mudp}}) = \gamma_{ik}^{dp} \sum_{m=1}^{K_{i}} p_{im}^{dl}$. Now, inserting \eqref{powerterm}, \eqref{interferenceZFmudp}, and \eqref{varZFmudp} into \eqref{SINRZFmudpaid} and utilizing that the pilot length is $\tau_{p}^{dp}$, the SE is obtained as given in \eqref{se_zf_mudp}.

\section*{Appendix B - Achievable SE with ZF-mucp}
Starting from \eqref{multicasttransmission} and applying \eqref{ZFMUCP} we have
\begin{align}
y_{ik} \!\! =&  (\hat{\mathbf{g}}_{ik}^{cp} - \tilde{\mathbf{g}}_{ik}^{cp})^{\!H} \! \sum_{j=1}^{G} \! \mathbf{w}_{j}^{\mathrm{ZF\!-\!mucp}} \! s_{j} \! + \! n \! \stackrel{(a)}{=} \!\!  \frac{\sqrt{\tau_{p}^{cp} p_{ik}^{u}} \beta_{ik}}{\tau_{p}^{cp} \! \sum_{m=1}^{K_{i}} \! p_{im}^{u} \beta_{im}} \!\! \sum_{j=1}^{G} \! \hat{\mathbf{g}}_{i}^{H} \! \mathbf{w}_{j}^{\mathrm{ZF\!-\!mucp}} \! s_{j} \! - \! \tilde{\mathbf{g}}_{ik}^{cpH} \! \sum_{j=1}^{G} \! \mathbf{w}_{j}^{\mathrm{ZF\!-\!mucp}} \! s_{j} \! + \! n \notag
\\
=& \underbrace{\dfrac{\sqrt{\tau_{p}^{cp} p_{ik}^{u}} \beta_{ik}}{\tau_{p}^{cp} \sum_{m=1}^{K_{i}}  p_{im}^{u} \beta_{im}} \sqrt{p_{i}^{dl} \gamma_{i} ( N-G)} }_{h} \underbrace{s_{i}}_{x} -  \underbrace{\tilde{\mathbf{g}}_{ik}^{cpH} \sum_{j=1}^{G} \mathbf{w}_{j}^{\mathrm{ZF-mucp}} s_{j}}_{v} + n \label{aidproofZFmudp}
\end{align}
where in $(a)$ we used $\hat{\mathbf{g}}_{jk}^{cp}  =  \dfrac{\sqrt{\tau_{p}^{cp} p_{jk}^{u}} \beta_{jk}}{ \tau_{p}^{cp} \sum_{k=1}^{K_{j}} p_{jk}^{u} \beta_{jk}} \hat{\mathbf{g}}_{j}$. 
Now applying Lemma \ref{MainLemma} considering $h$, $x$ and $v$ as shown in \eqref{aidproofZFmudp}, we obtain the effective SINR 
\begin{align}
\label{help_sinr_mucp}
\mathrm{SINR}_{ik}^{\mathrm{ZF-mucp}}   = \dfrac{\dfrac{\tau_{p}^{cp} p_{ik}^{u} \beta_{ik}^{2} p_{i}^{dl} \gamma_{i} ( N-G)}{(\tau_{p}^{cp} \sum_{m=1}^{K_{i}}  p_{im}^{u} \beta_{im})^{2}}}{1 + \sum_{j=1}^{G} \mathbb{E}[ |\tilde{\mathbf{g}}_{ik}^{cpH} \mathbf{w}_{j}^{\mathrm{ZF-mucp}} |^{2} ] }.
\end{align}
In the above equation for the terms $\mathbb{E}[ |\tilde{\mathbf{g}}_{ik}^{cpH} \mathbf{w}_{j}^{\mathrm{ZF-mucp}} |^{2} ]$ we have
\begin{align}
&\mathbb{E}[ |\tilde{\mathbf{g}}_{ik}^{cpH} \mathbf{w}_{j}^{\mathrm{ZF-mucp}} |^{2} ]
= \mathbb{E}[ \tilde{\mathbf{g}}_{ik}^{cpH} \mathbf{w}_{j}^{\mathrm{ZF-mucp}} \mathbf{w}_{j}^{\mathrm{ZF-mucp}H} \tilde{\mathbf{g}}_{ik}^{cp}]
= \mathrm{tr} ( \mathbb{E}[ \tilde{\mathbf{g}}_{ik}^{cp} \tilde{\mathbf{g}}_{ik}^{cpH} \mathbf{w}_{j}^{\mathrm{ZF-mucp}} \mathbf{w}_{j}^{\mathrm{ZF-mucp}H} ] ) \notag
\\
&\stackrel{(a)}{=} \mathrm{tr} ( \mathbb{E}[ \tilde{\mathbf{g}}_{ik}^{cp} \tilde{\mathbf{g}}_{ik}^{cpH} ] \mathbb{E}[ \mathbf{w}_{j}^{\mathrm{ZF-mucp}} \mathbf{w}_{j}^{\mathrm{ZF-mucp}H} ] ) = (\beta_{ik} - \gamma_{ik}) \mathrm{tr} (\mathbb{E}[ \mathbf{w}_{j}^{\mathrm{ZF-mucp}} \mathbf{w}_{j}^{\mathrm{ZF-mucp}H} ] ) \notag
\\
&= (\beta_{ik} - \gamma_{ik}) \mathbb{E}[ \mathbf{w}_{j}^{\mathrm{ZF-mucp}H} \mathbf{w}_{j}^{\mathrm{ZF-mucp}} ] = (\beta_{ik} - \gamma_{ik}) p_{j}^{dl} \label{interferenceZFmucp}
\end{align}
where (a) is due to the fact that $\tilde{\mathbf{g}}_{ik}^{cp}$ and $\hat{\mathbf{g}}_{i}$ are independent. Inserting \eqref{interferenceZFmucp} into \eqref{help_sinr_mucp} and noting that the pilot length is $\tau_{p}^{cp}$, we obtain \eqref{se_zf_mucp} for the SE of this UT.


\section*{Appendix C - MMF problem for MRT-undp}
First note that $\mathrm{SINR}_{ik}^{\mathrm{MRT-undp}}\!$, given in Proposition \ref{prop1}, is monotonically increasing with respect to $\gamma_{ik}^{dp}$, and also $\gamma_{ik}^{dp}$ is monotonically increasing with respect to $p_{ik}^{u}$. Therefore, the optimal value for $p_{ik}^{u}$ is $p_{ik}^{u*} = p_{ik}^{utot}$ and $\gamma_{ik}^{dp*}  =  \dfrac{\tau_{p}^{dp} p_{ik}^{utot} \beta_{ik}^{2} }{1 + \tau_{p}^{dp} p_{ik}^{utot} \beta_{ik}}$. Now we prove that at the optimal solution $\mathrm{SINR}_{ik}^{\mathrm{MRT-undp}} = \mathrm{SINR}_{jt}^{\mathrm{MRT-undp}} =\Gamma \quad \forall k,t,i,j$. Assume the contrary, i.e., that UT $t$ in group $j$ has the minimum SINR and there exists a UT $k$ in a group $i$ with $(i,k) \neq (j,t)$ such that $\mathrm{SINR}_{ik}^{\mathrm{MRT-undp}} > \mathrm{SINR}_{jt}^{\mathrm{MRT-undp}} $. Then one can improve $\mathrm{SINR}_{jt}^{\mathrm{MRT-undp}}$ by changing $p_{ik}^{dl}$ and $p_{jt}^{dl}$ respectively to $p_{ik}^{dl} - \delta$ and $p_{jt}^{dl} + \delta$, where $0<\delta<(\mathrm{SINR}_{ik}^{\mathrm{MRT-undp}} - \mathrm{SINR}_{jt}^{\mathrm{MRT-undp}} )\dfrac{1+\beta_{ik}P}{N \gamma_{ik}^{dp*}}$. Note that this just changes the $\mathrm{SINR}_{ik}^{\mathrm{MRT-undp}}$ and $\mathrm{SINR}_{jt}^{\mathrm{MRT-undp}}$, and the other SINRs remain intact. By performing this process once (or repeating it multiple times, if we have multiple UTs with same minimum SINR), we can increase the minimum SINR of the system, which contradicts our optimality assumption. Hence at the optimal solution all the SINRs are equal. Therefore, $p_{ik}^{dl*} = \dfrac{\Gamma (1+\beta_{ik} P)}{N \gamma_{ik}^{dp*}}$. Now by summing over all UTs in all groups and performing some straightforward operations we can find $\Gamma =  NP \big(\sum_{j=1}^{G} \sum_{t=1}^{K_{j}} \frac{1+\beta_{jt}P}{\gamma_{jt}^{dp*}} \big)^{-1}$.

\section*{Appendix D - MMF problem for ZF-mudp}
Starting from $\mathrm{SINR}_{ik}^{\mathrm{ZF-mudp}}$, given in Theorem \ref{TZFmudp}, and similar to Appendix C we can show that the optimal value for $p_{ik}^{u}$ is $p_{ik}^{u*} = p_{ik}^{utot}$ and $\gamma_{ik}^{dp*}  =  \dfrac{\tau_{p}^{dp} p_{ik}^{utot} \beta_{ik}^{2} }{1 + \tau_{p}^{dp} p_{ik}^{utot} \beta_{ik}}$. Now we prove that at the optimal solution $\mathrm{SINR}_{ik}^{\mathrm{ZF-mudp}} = \mathrm{SINR}_{jt}^{\mathrm{ZF-mudp}} = \Gamma \; \forall k,t,i,j$. Assume the contrary, i.e., that UT $t$ in group $j$ has the minimum SINR, and there exists a UT $k$ in a group $i$ with $(i,k) \neq (j,t)$ such that $\mathrm{SINR}_{ik}^{\mathrm{ZF-mudp}} > \mathrm{SINR}_{jt}^{\mathrm{ZF-mudp}}$. Denote $a_{ik} = (N - \nu_{i}) \gamma_{ik}^{dp*} p_{ik}^{dl}$ and $b_{ik} = 1+ \gamma_{ik}^{dp*} \sum_{m=1}^{K_{i}} p_{im}^{dl}   + P (\beta_{ik} - \gamma_{ik}^{dp*})$. Then one can increase the minimum SINR of the system by reducing $p_{ik}^{dl}$ to $p_{ik}^{dl} - \delta$, where $0 < \delta < \dfrac{a_{ik}b_{jt}-a_{jt}b_{ik}}{(N-\nu_{i}) \gamma_{ik}^{dp*} b_{jt} - a_{jt} \gamma_{ik}^{dp*}}$, which contradicts the assumption. Therefore at the optimal solution all UTs have the same SINR. Now consider UTs $k$ and $t$ in $i$th multicasting group. Let us denote $P_{i}^{dl} = \sum_{m=1}^{K_{i}} p_{im}^{dl}$, then we have
\begin{align}
\Gamma_{i} =  \dfrac{ \gamma_{ik}^{dp*} p_{ik}^{dl}}{1+ \gamma_{ik}^{dp*} P_{i}^{dl}   + P (\beta_{ik} - \gamma_{ik}^{dp*})} 
=  
\dfrac{ \gamma_{it}^{dp*} p_{it}^{dl}}{1+ \gamma_{it}^{dp*} P_{i}^{dl}   + P (\beta_{it} - \gamma_{it}^{dp*})} 
\end{align}
with $\Gamma = (N-\nu_{i}) \Gamma_{i}$. Hence we can write
\begin{align}
\label{pik_dl_zf_mudp_proof}
p_{ik}^{dl*} = \dfrac{\Gamma}{ (N-\nu_{i})} (\frac{1}{\gamma_{ik}^{dp*}} + P_{i}^{dl} + P \dfrac{\beta_{ik}}{\gamma_{ik}^{dp*}} - P ).
\end{align}
Summing over the downlink power of all UTs in group $i$ and after some straightforward operations we obtain
$P_{i}^{dl} =  \Gamma \Delta_{i}(N-\nu_{i} - \Gamma K_{i})^{-1}$, where $\Delta_{i} = \sum_{k=1}^{K_{i}} (\frac{1}{\gamma_{ik}^{dp*}} + P \dfrac{\beta_{ik}}{\gamma_{ik}^{dp*}} - P ) $. Note that as $\forall i \in \mathcal{G} \; P_{i}^{dl} \geq 0$ and $\sum_{k=1}^{K_{i}} P_{i}^{dl} = P_{dp} = P$, we have $\Gamma < \min_{ i \in \mathcal{G}} \{ \frac{N-\nu_i}{K_i} \}$. Summing over all groups downlink powers we have \eqref{sinr_equation} and $\Gamma$ can be found by solving it.

\section*{Appendix E - MMF problem for MRT-mucp}
First we prove that at the optimal solution $ \mathrm{SINR}_{jk}^{\mathrm{MRT-mucp}} = \mathrm{SINR}_{it}^{\mathrm{MRT-mucp}} \; \forall t,k,i,j$. Let us denote the user with the minimum SINR in $i$th group as $kmin_{i}$, i.e.,  $kmin_{i} = \argmin_{k \in \mathcal{K}_{i}}$. Now we prove that at the optimal solution of $\mathcal{P}^{\prime}2$ we have $ \mathrm{SINR}_{jkmin_{j}}^{\mathrm{MRT-mucp}} =  \mathrm{SINR}_{ikmin_{i}}^{\mathrm{MRT-mucp}} \; \forall i,j$. Assume the contrary, then $\exists j,i \in \mathcal{G}$ such that $ \mathrm{SINR}_{jkmin_{j}}^{\mathrm{MRT-mucp}} > \mathrm{SINR}_{ikmin_{i}}^{\mathrm{MRT-mucp}}$. Now one can change $p_{j}^{dl}$ and $p_{i}^{dl}$ respectively to $p_{j}^{dl} - \delta$ and $p_{i}^{dl} + \delta$ with $0 < \delta < ( \mathrm{SINR}_{jkmin_{j}}^{\mathrm{MRT-mucp}} - \mathrm{SINR}_{ikmin_{i}}^{\mathrm{MRT-mucp}} ) \dfrac{1+\beta_{jkmin_{j}}P}{N \gamma_{jkmin_{j}}^{cp}}$ and improve the minimum SINR of the system\footnote{If we have multiple groups with equal value of minimum SINR, we can improve the minimum SINR of the system by repeating the same procedure multiple times.}, which contradicts our optimality assumption. Now we prove that at the optimal solution the SINR of all the users within each group are the same, i.e., $\mathrm{SINR}_{ik}^{\mathrm{MRT-mucp}} = \mathrm{SINR}_{it}^{\mathrm{MRT-mucp}} \; \forall k,t \in \mathcal{K}_{i}, \forall i \in \mathcal{G}$. Assume the contrary, $\exists k,t \in \mathcal{K}_{i}$ such that $ \mathrm{SINR}_{ik}^{\mathrm{MRT-mucp}} > \mathrm{SINR}_{it}^{\mathrm{MRT-mucp}} $. Then one can improve the minimum SINR of this group by reducing $p_{ik}^{u}$ to $p_{ik}^{u} - \delta$, where $0 < \delta < \dfrac{(1+\tau_{p}^{cp} \sum_{m=1}^{K_{i}} p_{im}^{u} \beta_{im})(1+\beta_{ik}P)}{\tau_{p}^{cp} \beta_{ik}^{2} N p_{i}^{dl}} (\mathrm{SINR}_{ik}^{\mathrm{MRT-cp}} - \mathrm{SINR}_{it}^{\mathrm{MRT-cp}})$. Hence at the optimal answer for group $i$ we have
\begin{align}
\Phi_{i} = \dfrac{\gamma_{ik}^{cp}}{1+\beta_{ik}P} = \dfrac{\gamma_{it}^{cp}}{1+\beta_{it}P} \; \forall t,k \in \mathcal{K}_{i}, \forall i \in \mathcal{G}
\end{align}
where $\Phi_{i}$ is a fixed number. Equivalently we have
\begin{align}
\label{Upsiloneq}
\Upsilon_{i} =  \dfrac{p_{ik}^{u} \beta_{ik}^{2}}{1+\beta_{ik}P} = \dfrac{p_{it}^{u} \beta_{it}^{2}}{1+\beta_{it}P}  \; \forall k,t \in \mathcal{K}_{i}, \forall i \in \mathcal{G}
\end{align}
where $\Upsilon_{i}$ is a fixed constant. Considering the fact that $\mathrm{SINR}_{ik}^{\mathrm{MRT-mucp}}$ is strictly increasing with respect to $p_{ik}^{u}$ and noting that $p_{ik}^{u} \leq p_{ik}^{utot}$, the optimal uplink power will be equal to
\begin{align}
p_{ik}^{u*} =  \dfrac{1+\beta_{ik}P}{\beta_{ik}^{2}} \Upsilon_{i}  \; \forall k \in \mathcal{K}_{i}, \forall i \in \mathcal{G}
\end{align} 
where $\Upsilon_{i} = \min_{k \in \mathcal{K}_{i}} \dfrac{p_{ik}^{utot} \beta_{ik}^{2}}{1+\beta_{ik}P} $. Therefor $\mathrm{SINR}_{ik}^{\mathrm{MRT-mucp}} =  \Upsilon_{i} \dfrac{N p_{i}^{dl} \tau_{p}^{cp}}{1 + \tau_{p}^{cp} \sum_{m=1}^{K_{i}} p_{im}^{u} \beta_{im} }$. As we already showed the SINR at the optimal point is equal among all UTs and we have $\Gamma = \mathrm{SINR}_{ik}^{\mathrm{MRT-mucp}} \forall i,k$. Hence we have $p_{i}^{dl*} = \Gamma (1 + \tau_{p}^{cp} \sum_{m=1}^{K_{i}} p_{im}^{u*} \beta_{im}) / \tau_{p}^{cp} N \Upsilon_{i}$. Now summing $ p_{i}^{dl*}$ over all groups and employing the total available power constraint we achieve \eqref{MRTmucpSINR}.

\section*{Appendix F - MMF problem for ZF-mucp}
First we prove that at the optimal solution $ \mathrm{SINR}_{jk}^{\mathrm{ZF-mucp}} = \mathrm{SINR}_{it}^{\mathrm{ZF-mucp}} \; \forall t,k,i,j$. Let us denote the user with the minimum SINR in $i$th group as $kmin_{i}$, i.e.,  $kmin_{i} = \argmin_{k \in \mathcal{K}_{i}}$. Now we prove that at the optimal solution $ \mathrm{SINR}_{jkmin_{j}}^{\mathrm{ZF-mucp}} = \mathrm{SINR}_{ikmin_{i}}^{\mathrm{ZF-mucp}}$. Assume the contrary, then $\exists j,i \in \mathcal{G}$ such that $ \mathrm{SINR}_{jkmin_{j}}^{\mathrm{ZF-mucp}} > \mathrm{SINR}_{ikmin_{i}}^{\mathrm{ZF-mucp}}$. Now one can change $p_{j}^{dl}$ and $p_{i}^{dl}$ respectively to $p_{j}^{dl} - \delta$ and $p_{i}^{dl} + \delta$ with $0 < \delta < \big( \mathrm{SINR}_{jkmin_{j}}^{\mathrm{ZF-mucp}} - \mathrm{SINR}_{ikmin_{i}}^{\mathrm{ZF-mucp}} \big) \dfrac{1+(\beta_{jkmin_{j}}-\gamma_{jkmin_{j}}^{cp})P}{(N-G) \gamma_{jkmin_{j}}^{cp}}$ and improve the minimum SINR of the system, which contradicts our optimality assumption. Now we prove that at the optimal answer the SINR of all the UTs within each group are the same, i.e., $\mathrm{SINR}_{ik}^{\mathrm{ZF-mucp}} = \mathrm{SINR}_{it}^{\mathrm{ZF-mucp}} \; \forall k,t \in \mathcal{K}_{i}, \forall i \in \mathcal{G}$. Assume the contrary, $\exists k,t \in \mathcal{K}_{i}$ such that $ \mathrm{SINR}_{ik}^{\mathrm{ZF-mucp}} > \mathrm{SINR}_{it}^{\mathrm{ZF-mucp}} $. Then one can improve the minimum SINR of this group by reducing $p_{ik}^{u}$ to $p_{ik}^{u} - \delta$, where $0 < \delta < \dfrac{1+(\beta_{ik}-\gamma_{ik}^{cp})P}{p_{i}^{dl} (N-G)}(\mathrm{SINR}_{ik}^{\mathrm{ZF-mucp}} - \mathrm{SINR}_{it}^{\mathrm{ZF-mucp}})$. Hence at the optimal answer the SINR of all users within group $i$ are equal and we have
\begin{align}
\Delta_{i} = \dfrac{\gamma_{ik}^{cp}}{1+(\beta_{ik}-\gamma_{ik}^{cp})P} = \dfrac{\gamma_{it}^{cp}}{1+(\beta_{it} - \gamma_{it}^{cp})P} \; \forall t,k \in \mathcal{K}_{i}, \forall i \in \mathcal{G}.
\end{align}
Equivalently we have $\gamma_{ik}^{cp}(1+P\beta_{it}) = \gamma_{it}^{cp}(1+P\beta_{ik}) \; \forall t,k \in \mathcal{K}_{i}, \forall i \in \mathcal{G}$. Therefore 
\begin{align}
\Upsilon_{i} =  \dfrac{p_{ik}^{u} \beta_{ik}^{2}}{1+\beta_{ik}P} =  \dfrac{p_{it}^{u} \beta_{it}^{2}}{1+\beta_{it}P} \; \forall t,k \in \mathcal{K}_{i}, \forall i \in \mathcal{G}
\end{align}
where $\Upsilon_{i}$ is a fixed constant. Now note that it is exactly the same as \eqref{Upsiloneq} and hence the optimal uplink powers are given as 
\begin{align}
p_{ik}^{u*} =  \dfrac{1+\beta_{ik}P}{\beta_{ik}^{2}} \Upsilon_{i}  \; \forall k \in \mathcal{K}_{i}, \forall i \in \mathcal{G}
\end{align} 
where $\Upsilon_{i} = \min_{k \in \mathcal{K}_{i}} \dfrac{p_{ik}^{utot} \beta_{ik}^{2}}{1+\beta_{ik}P} $. Using the above result and after straightforward calculation we obtain $\Delta_{i} = \dfrac{\tau_{p}^{cp} \Upsilon_{i}}{1 + \tau_{p}^{cp}(E_{i} - P \Upsilon_{i})} \; \forall i \in \mathcal{G}$, where $E_{i} = K_{i} \Upsilon_{i} P + \Upsilon_{i} \sum_{m=1}^{K_{i}} \dfrac{1}{\beta_{im}}$. Since we proved that the SINR is equal for all UTs, we have $\Gamma = \mathrm{SINR}_{ik}^{\mathrm{ZF-mucp}} = (N-G) \Delta_{i} p_{i}^{dl}$, where $\Gamma$ is a fixed constant. Now, $p_{i}^{dl} = \dfrac{\Gamma}{(N-G) \Delta_{i}}$, and summing over all downlink powers and using the total available power constraint we achieve \eqref{ZFmucpSINR} and \eqref{ZFmucpDLpower} for the $\Gamma$ and $p_{i}^{dl*}$, respectively.

\bibliographystyle{IEEEtran}
\bibliography{IEEEabrv,multicastingMassive}

\end{document}